\documentclass[journal,10.5pt,draftclsnofoot,onecolumn]{IEEEtran}

\usepackage{cite}
\usepackage{amsmath,amssymb,amsfonts}
\usepackage{algorithmic}
\usepackage{graphicx}
\usepackage{textcomp}
\usepackage{xcolor}
\usepackage{subfigure}
\usepackage{amsthm}
\usepackage{mathrsfs}

\usepackage{pdfpages}
\usepackage{colortbl}
\definecolor{mygray}{gray}{0.85}
\usepackage{array}
\usepackage{caption}
\usepackage{setspace}
\usepackage{diagbox}
\usepackage[linesnumbered,boxed,ruled,commentsnumbered]{algorithm2e}%%算法包，注意设置所需可选项 algorithm \IncMargin{1em}
\usepackage[ruled,linesnumbered]{algorithm2e}

\newtheorem{theorem}{Theorem}{}
{}
\newtheorem{remark}{Remark}{}

\newtheorem{observation}[theorem]{Observation}
\newcolumntype{I}{!{\vrule width 1.25pt}}
\newlength\savedwidth

\newlength\savewidth
\newcommand\shline{\noalign{\global\savewidth\arrayrulewidth
		\global\arrayrulewidth 1.25pt}%
	\hline
	\noalign{\global\arrayrulewidth\savewidth}}

\usepackage{amsmath,bm}
\usepackage{cite}
\usepackage{amsmath,amssymb,amsfonts}
\usepackage{algorithmic}
\usepackage{amsthm}
\usepackage{graphicx}
\usepackage{textcomp}
\usepackage{xcolor}

\begin{document}

\title{\Huge Reconfigurable Intelligent Surface-Assisted Aerial-Terrestrial Communications \\via Multi-Task Learning}
%\title{When 6G meets RIS: Is there any need for reconfigurable intelligent surfaces?}

\author{
\IEEEauthorblockN{Xuelin Cao, Bo~Yang,~\IEEEmembership{Member,~IEEE},
 Chongwen Huang,~\IEEEmembership{Member,~IEEE}, 
 Chau Yuen,~\IEEEmembership{Fellow,~IEEE},
 Marco Di Renzo, \IEEEmembership{Fellow,~IEEE},  
 Dusit Niyato,~\IEEEmembership{Fellow,~IEEE},
 and Zhu Han,~\IEEEmembership{Fellow,~IEEE}}
 
%\thanks{The work of C. Yuen was supported by A*STAR under its RIE2020 Advanced Manufacturing and Engineering (AME) Industry Alignment Fund – Pre Positioning (IAF-PP) (Grant No. A19D6a0053). Any opinions, findings and conclusions or recommendations expressed in this material are those of the author(s) and do not reflect the views of A*STAR. The work of C. Huang was supported by the Fundamental Research Funds for the Central Universities. The work of Z. Han was partially supported by US Multidisciplinary University Research Initiative 18RT0073, NSF EARS-1839818, CNS1717454, CNS-1731424, and CNS-1702850.}
\thanks{X. Cao, B. Yang, and C. Yuen are with the Engineering Product Development Pillar, Singapore University of Technology and Design, Singapore 487372. (e-mails: xuelin$\_$cao, bo$\_$yang, yuenchau@sutd.edu.sg).}
\thanks{C. Huang is with College of Information Science and Electronic Engineering, Zhejiang University, Hangzhou 310027, China, and with International Joint Innovation Center, Zhejiang University, Haining 314400, China, and also with Zhejiang Provincial Key Laboratory of Info. Proc., Commun. $\&$ Netw. (IPCAN), Hangzhou 310027, China. (e-mail: chongwenhuang@zju.edu.cn).}
\thanks{M. Di Renzo is with Universit\'e Paris-Saclay, CNRS, CentraleSup\'elec, Laboratoire des Signaux et Syst\`emes, 3 Rue Joliot-Curie, 91192 Gif-sur-Yvette, France. (e-mail: marco.di-renzo@universite-paris-saclay.fr).}
\thanks{D. Niyato is with the School of Computer Science and Engineering, Nanyang Technological University, Singapore. (e-mail: dniyato@ntu.edu.sg).}
\thanks{Z. Han is with the Department of Electrical and Computer Engineering in the University of Houston, Houston, TX 77004 USA, and also with the Department of Computer Science and Engineering, Kyung Hee University, Seoul, South Korea, 446-701. (e-mail: zhan2@uh.edu).}
}
\maketitle

\begin{abstract}
The aerial-terrestrial communication system constitutes an efficient paradigm for supporting and complementing terrestrial communications. However, the benefits of such a system cannot be fully exploited, especially when the line-of-sight (LoS) transmissions are prone to severe deterioration due to complex propagation environments in urban areas. The emerging technology of reconfigurable intelligent surfaces (RISs) has recently become a potential solution to mitigate propagation-induced impairments and improve wireless network coverage. Motivated by these considerations, in this paper, we address the coverage and link performance problems of the aerial-terrestrial communication system by proposing an RIS-assisted transmission strategy. In particular, we design an adaptive RIS-assisted transmission protocol, in which the channel estimation, transmission strategy, and data transmission are independently implemented in a frame. On this basis, we formulate an RIS-assisted transmission strategy optimization problem as a mixed-integer non-linear program (MINLP) to maximize the overall system throughput. We then employ multi-task learning to speed up the solution to the problem. Benefiting from multi-task learning, the computation time is reduced by about four orders of magnitude. Numerical results show that the proposed RIS-assisted transmission protocol significantly improves the system throughput and reduces the transmit power. 
\end{abstract}


\IEEEpeerreviewmaketitle

\section{Introduction}

With the emergence of space-air-ground integrated networks, aerial-terrestrial communications play an increasingly significant role in providing ubiquitous coverage and offering flexible end-to-end services \cite{liu2018space, zhang2017software, kato2019optimizing}. Unmanned aerial vehicles (UAVs), commonly known as aerial platforms, have been widely leveraged to complement and strengthen existing terrestrial networks. Compared to traditional terrestrial communications, UAV-enabled aerial-terrestrial communications have attracted increasing interest due to their prominent features of fully controllable mobility, line-of-sight (LoS) transmission, and cost-effectiveness \cite{zeng2016wireless,wu2019fundamental,cui2018robust}. In particular, UAVs are gaining popularity in various applications such as search and rescue, cargo/packet delivery, communication platforms, precise agriculture, etc. In general, these promising UAV-enabled applications can be classified into three types: UAV-enabled ubiquity coverage, UAV-enabled relaying, and UAV-enabled information dissemination/data collection \cite{zeng2017energy}. This paper mainly focuses on UAV-enabled ubiquity coverage and information dissemination/data collection, where the UAV-enabled system is employed to provide seamless wireless coverage within its serving area.

However, considering the complicated and unpredictable propagation environment in urban areas, conventional UAV-enabled aerial-terrestrial communications may face several challenges. The first prominent challenge is the blockage of LoS links due to obstacles that may lead to coverage and connectivity problems. Another critical challenge is the non-stationary nature of aerial-terrestrial communications because of the high mobility of UAVs. In addition, the severe path loss caused by long-distance transmissions has a significant impact on the transmission rate. Recently, the technology of reconfigurable intelligent surfaces (RISs) has been proved as a potential solution for tackling the above challenges in wireless communications. Specifically, an RIS consists of a large number of passive reflecting elements that can be controlled to adjust the amplitude and$/$or phase of the incident signals, thus dynamically reconfiguring the wireless propagation environment to provide desirable propagation properties and diverse transmission channels \cite{di2020hybrid,wu2019towards,9140329,huang2020holographic,di2019smart,add4}. Thanks to the promising characteristics of RISs, we explore an RIS-assisted aerial-terrestrial communication system where the RIS elements can be dynamically optimized for performance improvement. Furthermore, an efficient RIS-assisted transmission protocol based on multi-task learning for application to aerial-terrestrial communications is designed to achieve higher data rate, higher flexibility, and more reliable transmission and coverage.

\subsection{Related Works}
The increasing demand for broadband wireless communications has led to the enormous success of terrestrial networks. However, the high unpredictability of wireless environments poses serious challenges to the deployment of terrestrial networks. Aerial-terrestrial communications offer the flexibility to overcome some of the limitations of terrestrial communications, e.g., ubiquitous coverage\cite{zeng2019accessing}. In this context, UAVs can be employed as aerial base stations (BSs), access points, or relays, to assist terrestrial communications\cite{zeng2018cellular}. In particular, UAV-enabled aerial-terrestrial communications in which the UAVs are used as BSs are envisioned as a promising solution to complement terrestrial communications \cite{lin2018sky,fotouhi2019survey,zhang2019cellular}. Specifically, Wu \textit{et al.} exploited the UAV mobility to achieve user fairness \cite{wu2018joint}, and to improve the total capacity of UAV-enabled multi-user communication systems \cite{wu2018capacity}. Zeng \textit{et al.} considered the energy-efficiency and UAV mission completion time of UAV-enabled aerial-terrestrial communications \cite{zeng2017energy,zeng2018trajectory}. Mohammad \textit{et al.} investigated the mission completion time of UAV as flying BSs \cite{8053918} and proposed a tractable analytical framework for the coverage and rate analysis \cite{7412759}. Boris \textit{et al.} analyzed the impact of the UAV altitude that is affected by the UAV channel and directional antenna \cite{galkin2017coverage}. Pang \textit{et al.} employed UAVs to collect data and recharge sensor nodes \cite{pang2014efficient}. You \textit{et al.} investigated the data collection efficiency of UAV-enabled wireless sensor networks \cite{you20193d}. Gong \textit{et al.} solved the optimal UAV flight time for interval data collection \cite{gong2018flight}. Chen \textit{et al.} discussed the deployment problem of cache-enabled UAV in a cloud radio access network \cite{chen2017caching}. 

In the context of RIS-aided communications, Wu \textit{et al.} studied the problem of joint active and passive beamforming \cite{wu2018intelligent}. To achieve high energy efficiency, Huang \textit{et al.} investigated an RIS-assisted downlink multi-user system by joint optimizing the transmit power and the passive beamforming \cite{huang2019reconfigurable,huang2018energy}. To increase the sum-rate, Guo \textit{et al.} studied an RIS-aided multi-user multiple-input single-output downlink system by jointly designing the beamforming and RIS phase shifts \cite{guo2020weighted}. To assess the effect of RIS phase shifts on the data rate, Zhang \textit{et al.} derived the required number of phase shifts for RIS-assisted communication systems \cite{zhang2020reconfigurable}. Yang \textit{et al.} proposed a practical transmission protocol for RIS-enhanced orthogonal frequency division multiplexing (OFDM) system \cite{yang2020intelligent}, and Zheng \textit{et al.} extended these investigations to reduce the channel estimation overhead \cite{zheng2019intelligent}. Wei \textit{et al.} analyzed the RIS channel estimation in terms of accuracy and overhead \cite{wei2020channel}. To improve the channel estimation accuracy, You \textit{et al.} designed an RIS training reflecting matrix for a single-user communication system \cite{you2020channel}. Abeywickrama \textit{et al.} investigated an RIS system by considering a practical model for the phase shift\cite{abeywickrama2020intelligent}. Yang \textit{et al.} \textit{et al.} analyzed an RIS-assisted uplink multi-user system \cite{yang2019irs}, and Wang \textit{et al.} explored RIS-aided IoT networks with uplink over-the-air computation and downlink energy beamforming \cite{ZWang}. The physical-layer security of RIS-assisted systems was analyzed in \cite{yu2019enabling}, and MAC layer protocol for RIS-assisted multi-user system was designed in \cite{XCao}. Hu \textit{et al.} designed an RIS-assisted sensing system for posture recognition \cite {hu2020reconfigurable}. Huang \textit{et al.} explored the potential of deep reinforcement learning in RIS-aided systems \cite{9110869}, and Yang \textit{et al.} introduced federated learning into RIS-aided systems \cite{add2}. Deep learning technologies instead of conventional optimization methods were investigated in \cite{yang2020computation,add3}.

At the time of writing, several works have explored the application of RISs in UAV communications to improve energy efficiency by jointly optimizing the UAV’s trajectory and RIS passive beamforming \cite{long2020reflections,li2020reconfigurable,ge2020joint,add7}. However, the RIS-aided transmission protocol design and the joint optimization of RIS resource allocation and configuration in aerial-terrestrial communications are still open issues.

\subsection{Contributions and Organizations}

In this paper, the RIS is used to support multiple UAV-user pairs for their communication links improvement based on the proposed RIS-aided transmission protocol. In particular, we consider two issues in this paper: 1) how to allocate the RIS elements  (i.e., whether or not to allocate and how many to allocate); and 2) how to configure the RIS reflecting coefficients. To tackle these two issues, we formulate an RIS-assisted transmission strategy optimization problem, which is solved via multi-task learning more efficiently to maximize the overall system capacity. The main contributions of this paper are summarized as follows.

\begin{itemize}
\item We propose an RIS-assisted transmission strategy to allocate the RIS elements and configure the RIS reflecting coefficients for multiple UAV-user pairs. By executing the proposed strategy, the RIS controller decides for each UAV-user pair whether to communicate via the RIS, thereby enhancing the coverage and propagation quality of aerial-terrestrial communications.   

\item We design an adaptive RIS-assisted transmission protocol, which is a frame-based periodic structure. Channel estimation, transmission strategy, and data transmissions are alternately executed in a frame. This design has four major benefits: 1) achieving distributed aerial-terrestrial communications; 2) adapting for the dynamical wireless environment; 3) efficient utilization of the RIS, and 4) interference reduction.

\item We introduce a deep neural network-based machine learning model, called multi-task learning, to infer the optimal transmission strategy accurately in near-real-time. Compared with conventional analytical and numerical methods, the proposed approach can speed up the problem solution, reduce the computation complexity, and improve the solution accuracy. We then evaluate the proposed protocol in terms of signal-to-noise ratio (SNR), transmit power, and overall system throughput. We also discuss the accuracy and time cost for solving the problem via multi-task learning.  
\end{itemize}    			

The remainder of this paper is organized as follows. Section \ref{sec2} introduces the RIS-assisted aerial-terrestrial communication system model. Section \ref{sec3} describes the RIS-assisted transmission protocol. In Section \ref{sec4}, we formulate the transmission strategy optimization problem for RIS configuration. Section \ref{sec5}  solves the formulated mixed-integer non-linear program (MINLP) problem via multi-task learning. Simulation results and discussions are provided in Section \ref{sec6}. Section \ref{sec7} concludes the paper.
     
\textit{Notations}: Boldface letters denote column vectors or matrices. $\mathbb{C}^{x\times y}$ represents the space of $x \times y$ complex-valued matrices. $\text{diag}(x)$ returns a square diagonal matrix with the elements in $x$ on its main diagonal. $(\cdot)^T$ and $(\cdot)^H$ denote the transpose and the conjugate transpose, respectively. $\vert \cdot \vert$ denotes the absolute value of a complex number. %Additionally, the main notations used throughout the paper and their physical meanings are summarized in Table \ref{Abb}. 

\section{System Model}\label{sec2}  
We consider an RIS-assisted multi-user downlink aerial-terrestrial communication system, as illustrated in Fig. \ref{sys}. In this system, an RIS is attached to a building facade to assist the downlink communications of $K$ UAV-user pairs, denoted by the set $\mathcal{K}\!=\!\{1, 2,\! \ldots\!, K\}$, and each UAV-user pair is equipped with a single antenna. We assume that the RIS is equipped with $N$ passive and low-cost reflecting elements denoted by the set $\mathcal{N}\!=\!\{1, 2, \!\ldots\!, N\}$. The $N$ reflecting elements of the RIS are arranged in a uniform planar array (UPA) configuration \cite{di2020hybrid}. The RIS operates as a nearly passive device, and in particular, only the RIS controller and configuration circuitry consume power. A quasi-static fading channel model is considered, in which UAV-user pairs' locations are kept fixed within a frame.

\begin{figure}[t]
\small	
\centering{\includegraphics[height=2.2in,width=2.9in]{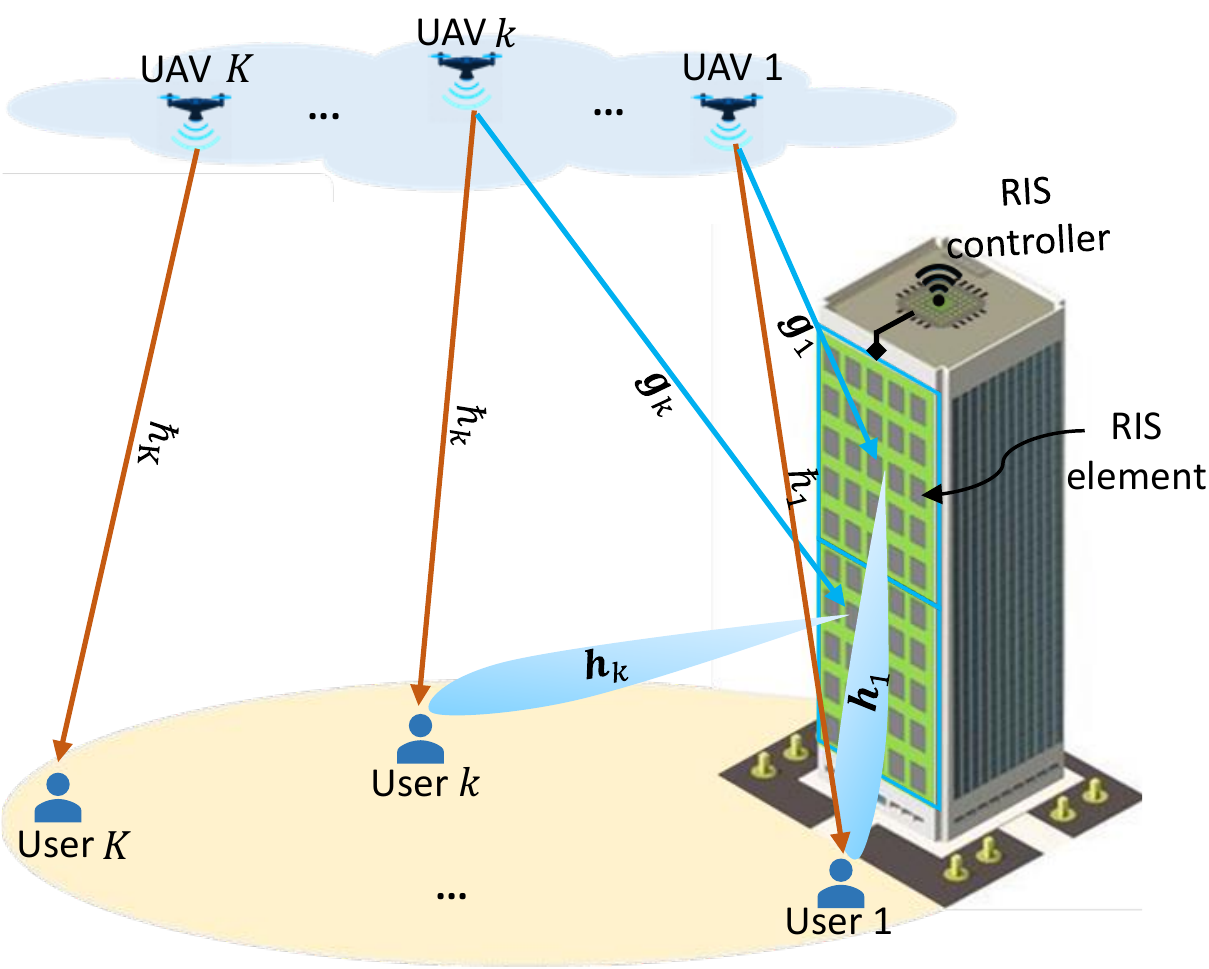}}
\caption{An RIS-assisted multi-user downlink aerial-terrestrial communication system.} 
\label{sys}
\end{figure}

Due to the limited number of RIS elements in practice, it is difficult for the RIS to serve all UAV-user pairs. In this case, the RIS controller needs to allocate the RIS elements to some UAV-user pairs. In each frame, we assume that $N$ RIS elements are equally divided into $L$ RIS groups, denoted by the set $\mathcal{L}\!=\!\{1,2,\ldots,L\}$. Let $\mathbf{L}\! = \!\{\mathbf{l}_1, \mathbf{l}_2, \ldots, \mathbf{l}_L\}$ denote the set of $L$ RIS groups in a frame, where $\mathbf{l}_l\!=\!\{l_1, l_2, \ldots, l_{N/L}\}, \forall l \in \mathcal{L}$ denotes the $l$-th RIS group comprises of $N/L$ elements. Note that the RIS group size and state vary with the frame, while they keep unchanged within a frame. Moreover, the frequency channel is divided into $C$ orthogonal sub-carriers. Since different UAV-user pairs use different sub-carriers to communicate, the signal on a specific sub-carrier can be decoded at the receiver, although the RIS will scatter all signals that reach it. This implies that one RIS group can serve one UAV-user pair on one sub-carrier without inter-group interference. For simplicity, the amplitude and phase responses are assumed to be zero outside the narrow band of a sub-carrier \cite{yang2020intelligent,zheng2019intelligent}. The total bandwidth $B$ is divided into $\omega_1B$ and $\omega_2B$, where $\omega_1+\omega_2=1$. Specifically, the bandwidth $\omega_1B$ is used for the RIS-assisted transmissions, while the bandwidth $\omega_2B$ is used for the transmissions that are not assisted by the RIS. The links from $K$ UAVs to the RIS are denoted by the set $\mathbf{G} = \{\mathbf{g}_1, \mathbf{g}_2, \ldots, \mathbf{g}_K\}$, and the links from the RIS to $K$ users are denoted by the set $\mathbf{H} = \{\mathbf{h}_1, \mathbf{h}_2, \ldots, \mathbf{h}_K\}$, where $\mathbf{g}_k\in \mathbb{C}^{(N/L)\times1}, k\in \mathcal{K}$ and $\mathbf{h}_k\in \mathbb{C}^{1\times(N/L)}, k\in \mathcal{K}$. The direct links of $K$ UAV-user pairs are denoted by $\mathbf{r} = \{\hslash_1, \hslash_2, \ldots, \hslash_K\}$. In the considered model, each RIS group and each sub-carrier is used to assist one UAV-user pair at a time. Otherwise, collisions occur at the user side.

\subsection{Channel Model}\label{sec2a}
As shown in Fig. \ref{sys}, the UAV-user pair's transmission channels are composed of the direct path (i.e., the UAV-user link) and the reflect paths (i.e., the UAV-RIS links and the RIS-user links). On this basis, the Rayleigh fading channel model is considered for the UAV-user link and the UAV-RIS links, the Rician fading channel model is considered for the RIS-user links \cite{long2020reflections,li2020reconfigurable,ge2020joint}. Therefore, the received signal at the $k$-th user is denoted by
\begin{equation}\label{received signal1}
\centering
y_k = \underbrace{\hslash_ks_k}_{\text{direct path}}+\underbrace{\mathbf{h}_k \mathbf{\Theta}_k \mathbf{g}_k s_k}_{\text{reflect paths}}+w_k, 
\end{equation}
where $s_k$ represents the signal of the $k$-th UAV, which is identically distributed (i.i.d.) random variable with zero mean and unit variance, $w_k$ is additive white Gaussian noise (AWGN) at the $k$-th user with zero mean and variance $\sigma^2$. $\mathbf{\Theta}_k$ is the RIS reflection coefficient matrix of the $k$-th UAV-user pair, which can be expressed as
\begin{equation}\label{channel1}
\centering
\mathbf{\Theta}_k =\text{diag}(\phi_k^{l_1},\ldots,\phi_k^{l_n},\ldots,\phi_k^{l_{N/L}}),
\end{equation}
where $\!\phi_k^{l_n}\!=\!\beta_k^{l_n}e^{j\theta_k^{l_n}}$ is the reflection coefficient of element $n$ on the $l$-th RIS group for the $k$-th UAV-user pair, ${l_n}\in \mathbf{l}_l$ denotes element $n$ on the $l$-th RIS group, $\{\theta_k^{l_n},\beta_k^{l_n}\}$ are the phase shift and amplitude reflection coefficient of element $n$ on the $l$-th RIS group for the $k$-th UAV-user pair. In practice, we assume that a continuous phase shift with a constant amplitude reflection coefficient is applied to each RIS element, i.e., $\vert\beta_k^{l_n}\vert=1, \theta_k^{l_n} \in [0, 2\pi), \forall {l_n}\in \mathbf{l}_l$. Let $\mathbf{\Psi}=\{\bm{\theta}_1,\bm{\theta}_2,\ldots, \bm{\theta}_K\}$ denote an $(N/L)*K$ dimensional RIS phase shift matrix, where $\bm{\theta}_k\!=\!\{\theta_k^{l_1},\theta_k^{l_2},\ldots, \theta_k^{l_{N/L}}\}$ is an $N/L$ dimensional phase shift vector that are aligned to the $k$-th UAV-user pair. 

The direct channel gain of the $k$-th UAV-user link, $\hslash_k$, is expressed as
\begin{equation}\label{LoS}
\centering
\hslash_k = \sqrt{h_0d_{k,Uu}^{-\tau_{k,Uu}}}\bar{h}, 
\end{equation}
where $h_0$ is the path loss at a reference distance $d_0 = 1$ m, $d_{k,Uu}$ is the LoS distance of the $k$-th UAV-user pair, $\tau_{k,Uu}\geq 2$ is the path loss exponent of the $k$-th UAV-user link, and $\bar{h}$ is the random scattering component, which is a zero-mean and unit-variance circularly symmetric complex Gaussian (CSCG) random variable.

The channel gain vector of the $k$-th UAV-RIS links, $\mathbf{g}_k$, is given by
\begin{equation}
\centering
{\mathbf g}_k = \sqrt{h_0d_{k,UR}^{-\tau_{k,UR}}} \mathbf{\bar g}_k, 
\end{equation}
where $d_{k,UR}$ is the distance between the $k$-th UAV and the RIS. $\tau_{k,UR}\geq 2$ is the path loss exponent of the $k$-th UAV-RIS links. $\mathbf{\bar g}_k$ is the array responses, which is denoted by
\begin{equation}
\centering
\mathbf{\bar g}_k = {\mathbf a}_{R}\left( \varphi^{AOA}_{k},\vartheta^{AOA}_{k}\right), 
\end{equation}
$\varphi^{AOA}_k$ $\left(\vartheta^{AOA}_k\right) $ is the corresponding azimuth (elevation) angle-of-arrival (AoA) of the $k$-th UAV-RIS links. ${\mathbf a}_{R}\left( \varphi,\vartheta\right)\triangleq\left[1,\ldots,e^{j\frac{2\pi}{\lambda}d\left(l_x\sin\varphi\sin\vartheta+l_y\cos\vartheta\right) },\ldots,\right]^T$, where $l_x$ and $l_y$ ($0\le\{l_x,l_y\}\le N/L-1$) are the length and width of RIS group $l$, $d$ is the antenna separation and $\lambda$ is the carrier wavelength.

In the considered system model, both LoS and NLoS components are considered in the RIS-user $k$ links, the channel gain vector of the RIS-user $k$ links is denoted by $\mathbf{h}_k$ and is given by
\begin{equation}
\centering
\mathbf{h}_k =  \sqrt{h_0d_{k,Ru}^{-\tau_{k,Ru}}}\left(\sqrt{\frac{\alpha}{1+\alpha}}\bar{\mathbf{h}}_k+\sqrt{\frac{1}{1+\alpha}}\hat{\mathbf{h}}_k\right), 
\end{equation}
where $d_{k,Ru}$ is the distance between the RIS and the $k$-th user. $\tau_{k,Ru}\geq 2$ is the path loss exponent of the RIS-user $k$ links. $\alpha$ is the Rician factor, and $\bar{\mathbf{h}}_k$ is the LoS components vector, which is given by
\begin{equation}
\centering
\mathbf{\bar h}_k = {\mathbf a}_{R}^H\left( \varphi^{AOD}_{k},\vartheta^{AoD}_{k}\right), 
\end{equation}
where $\varphi^{AOD}_k$ $\left(\vartheta^{AOD}_k\right) $ is the corresponding azimuth (elevation) angle-of-departure (AoD) of the RIS-user $k$ links. Also, $\mathbf{\hat h}_{k}$ is the NLoS components vector, which follows i.i.d. complex Gaussian distributed with zero mean and unit variance.

\subsection{Communication Model}\label{sec2b}
Let $\mathcal{U}=\{{u}_{1},{u}_{2},\ldots,{u}_{K}\}$ denote the $K$-dimensional RIS group occupation vector of $K$ UAV-user pairs, where ${u}_{k}$ represents that the $k$-th UAV-user pair communicates via the $l$-th RIS group or without the assistance of the RIS. Then, we have
\begin{equation}
\centering
 u_{k} =
\left\{\begin{array}{l}
l,  \;\;\;\;\;\text{if UAV}\ k\ \text{transmits via RIS group}\ l,  \\
0,  \;\;\;\;\;\text{if UAV}\ k\ \text{transmits without the RIS}.
\end{array} \right.
\end{equation}
Since one RIS group serves one UAV-user pair at a time, then we have 
\begin{equation}
\centering
u_k \neq u_{k'}, \exists u_k\neq 0, u_k'\neq 0, \forall k, k'\in \mathcal{K}.   
\end{equation} 
In addition, we define $\mathcal{F}=\{f(u_1),f(u_2),\ldots,f(u_K)\}$ as the $K$-dimensional RIS-assisted transmission decision vector of $K$ UAV-user pairs, where $f(u_k)$ represents that the $k$-th UAV-user pair is assigned one RIS group or not. Then, we have 
\begin{equation}
\centering
 f(u_{k}) =
\left\{\begin{array}{l}
1,\;\;\;\;\;\;\;\;\;u_k\neq 0,  \\
0,\;\;\;\;\;\;\;\;\;u_k=0.
\end{array} \right.
\end{equation}
Since one UAV-user pair either to communicate without the assistance of the RIS or to communicate with the assistance of only one RIS group, we have $\sum_{k=1}^{K} f(u_{k}) = L$.

Moreover, we define $\mathcal{P} = \{p_1, p_2, \ldots, p_K\}$ as the $K$-dimensional RIS element occupation ratio vector of $K$ UAV-user pairs, where $p_k$ is the RIS element occupation ratio of the $k$-th UAV-user pair. Then, we have
\begin{equation}\label{pk}
\centering 
p_{k} =
\left\{\begin{array}{l}
\frac{1}{\sum_{k=1}^{K} f(u_{k})},\;\;\;\;\;u_k\neq 0, \vspace{1ex} \\
0,\;\;\;\;\;\;\;\;\;\;\;\;\;\;\;\;\;\;\;\; u_k= 0.
\end{array} \right.
\end{equation}

Let $\mathcal{C}=\{{c}_{1},{c}_{2},\ldots,{c}_{K}\}$ be the $K$-dimensional sub-carrier bandwidth occupation ratio vector of $K$ UAV-user pairs, where ${c}_{k} \in (0, 1)$ is the occupation ratio of the bandwidth of the $k$-th UAV-user pair, which is calculated as
\begin{equation}\label{ck}
\centering 
c_{k} =
\left\{\begin{array}{l}
p_k\omega_1,\;\;\;\;\;\;\;\;\; \;\;\;\;\;\;\;\;\;\;u_k\neq 0, \vspace{1ex} \\
\frac{\omega_2}{{K\!-\sum_{k=1}^{K} f(u_{k})}},\;\;\;\;\;\; u_k= 0.
\end{array} \right.
\end{equation}

Given the channel model in Section \ref{sec2a}, let $\rho^2$ be the transmitter power of the $k$-th UAV. Accordingly, the received SNR at the $k$-th user is calculated as 
\begin{equation}\label{SNRE1}
\centering 
\text{SNR}_k =
\left\{\begin{array}{l}
\left|\left(\hslash_k+\mathbf{h}_k \mathbf{\Theta}_k \mathbf{g}_k\right)\rho\right|^2/\sigma^2,\;\;\;\;\;u_k\neq 0,\vspace{1ex}  \\
\left|\hslash_k\rho\right|^2/\sigma^2,\;\;\;\;\;\;\;\;\;\;\;\;\;\;\;\;\;\;\;\;\;\;\;\;\;\; u_k = 0.
\end{array} \right.
\end{equation}

\subsection{Computation Model}
In the considered computation model, the transmission strategy of the $k$-th UAV-user pair is denoted as $\mathcal{D}_k\!=\!\{u_k, \bm{\theta}_k\}$, where $u_k$ values zero or nonzero, indicating whether the $k$-th UAV-user pair communicates via the $l$-th RIS group or without the assistance of the RIS, and $\bm{\theta}_k$ is a continuous value, representing the RIS phase shift vector of the $k$-th UAV-user pair. Therefore, the transmission strategy of $K$ UAV-user pairs is denoted by 
\begin{equation}
\centering
\mathbf{D}=\{\mathcal{D}_1,\ldots, \mathcal{D}_K\}. 
\end{equation} 

\subsection{Power Consumption Model}
The power consumption of the $k$-th UAV-user pair, denoted by $P_k$, is formulated as
\begin{equation}\label{Pc}
\centering 
P_k =
\left\{\begin{array}{l}
\mu\rho^2+P_{k, U}+P_{k, u}+P_{k,R},\;\;\;\;\;\;\;\;\;\;\;\; u_k\neq 0, \vspace{1ex} \\
\mu\rho^2+P_{k, U}+P'_{k, u},\;\;\;\;\;\;\;\;\;\;\;\;\;\;\;\;\;\;\;\;\;\;\;\; u_k = 0,
\end{array} \right.
\end{equation} 
where $\mu=\nu^{-1}$ with $\nu$ being the power amplifier efficiency, $P_{k, U}$ is the static hardware power consumption at the $k$-th UAV, $P_{k, u}$ and $P'_{k, u}$ denote the static hardware power consumption at the $k$-th user via the RIS or not, respectively. Also, $P_{k,R}$ is the power consumption at the RIS group for serving the $k$-th UAV-user pair, which is given by 
\begin{equation}\label{Pc1}
\centering 
P_{k,R} = p_kP_R,
\end{equation} 
where $P_{R}$ is the power consumption at the RIS. Note that $P_{k,R} = 0$ as $u_k=0$.

Therefore, the overall power consumption of $K$ UAV-user pairs can be expressed as 
\begin{equation}
P_o\!=\!K\mu\rho^2\!+\!\sum_{k=1}^{K}\!\left(P_{k,U}\!+\!f(u_k)\!P_{k, u}\!+\!\left(1\!-\!f(u_k)\right)\!P'_{k, u}\!+\!p_kP_R\right).
\end{equation}

\begin{figure*}[t]
\small
 \centering
          	\subfigure[The RIS-assisted transmission protocol design]{
    		\begin{minipage}[b]{1\textwidth}
    			\centering
    			 \includegraphics[height=2.2in,width=6.2in]{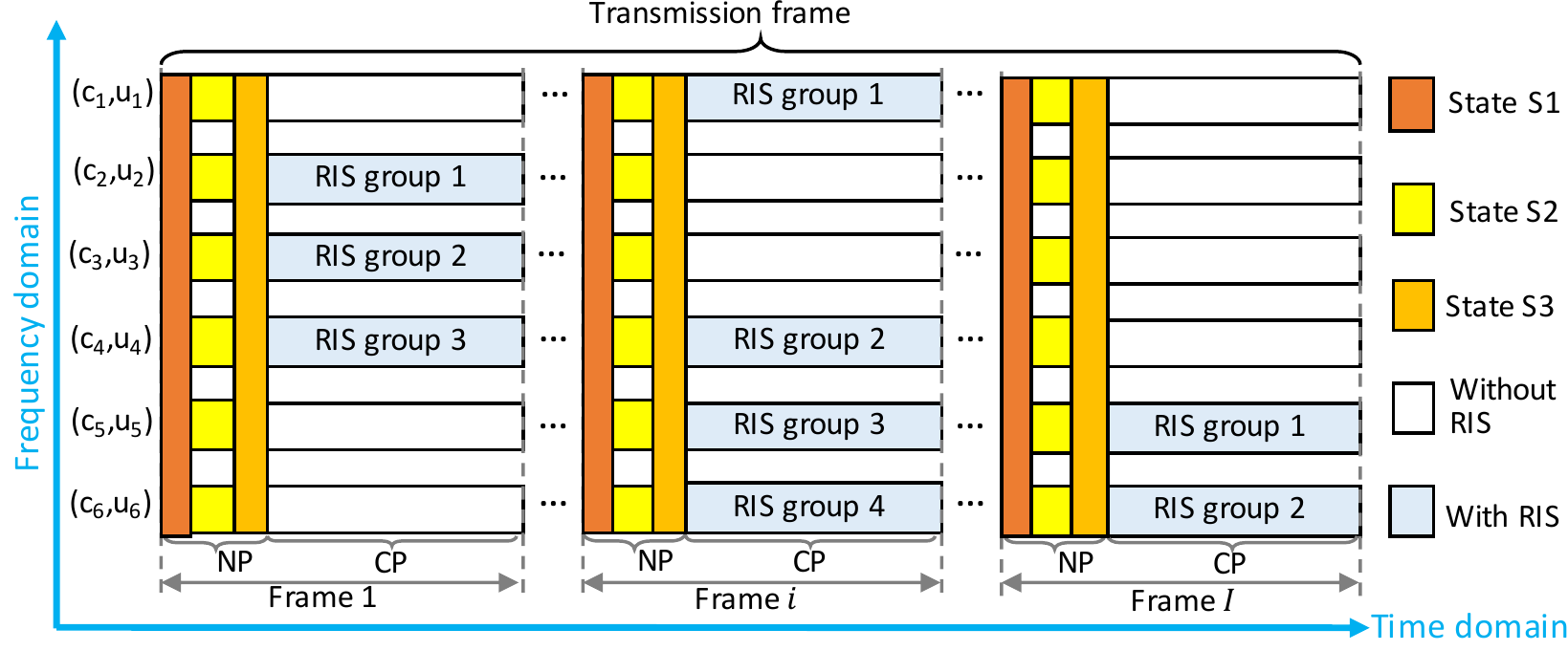}
    		\end{minipage}
    	}
    	\hfill
    	\subfigure[Three states in each NP for synchronization, channel estimation, and optimization]{
    		\begin{minipage}[b]{1\textwidth}
    			\centering
    		 \includegraphics[height=1.5in,width=6.2in]{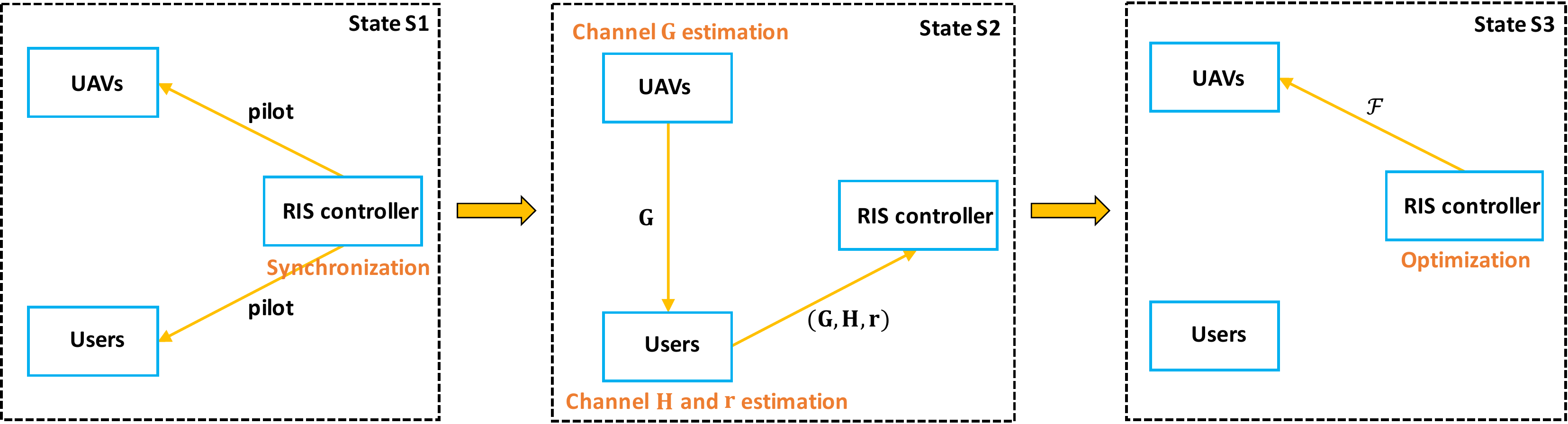}
    		\end{minipage}
    	}
    	\caption{RIS-assisted transmission protocol with three states and data communications, where $K=6$.} 
    	\label{P}
    \end{figure*}

\section{RIS-Assisted Transmission Protocol}\label{sec3}
The proposed RIS-assisted aerial-terrestrial transmission protocol is shown in Fig. \ref{P} (a). In the frequency domain, each UAV-user pair occupies a sub-carrier. In the time domain, the transmission frame can be divided into $I$ frames. By denoting $\mathcal{I}=\{1,2,\ldots,I\}$, frame $i$ ($i\in\mathcal{I}$) consists of two phases: the negotiation phase (NP) and the communication phase (CP). Furthermore, each NP consists of three states: synchronization (state S1), channel estimation (state S2), and optimization  (state S3), as illustrated in Fig. \ref{P} (b). At state S1, the RIS controller sends the pilot across the whole frequency band for synchronization and following channel estimation. At state S2, $K$ UAVs estimate the channel information $\mathbf{G}$, and $K$ users estimate the channel information $\mathbf{H}$ and $\mathbf{r}$, then $K$ users feedback the channel information $\mathbf{G}$, $\mathbf{H}$, and $\mathbf{r}$ to the RIS controller. At state S3, the RIS controller optimizes the transmission strategy and initiates the transmissions of $K$ UAV-user pairs. These three states are independent and their duration are given by the set $\mathbf t=\{t_{s1},t_{s2},t_{s3}\}$. As the CP arrives, the UAV-user pairs that are assigned with RIS elements start their transmissions via the RIS, while the other UAV-user pairs that are not assigned with the RIS elements communicate only via the direct link. From a practical implementation perspective, the challenges faced in the proposed protocol focus on the RIS channel estimation and the RIS beamforming design \cite{add5,add6}.

\subsection{Channel Estimation}\label{sec3a}
In the proposed protocol, the channel information (i.e., $\mathbf{G}$, $\mathbf{H}$, and $\mathbf{r}$) can be estimated at the UAVs and users in a frame-by-frame basis. Here, we assume that each UAV-user pair can obtain perfect channel estimation \cite{9133107}. By considering the path loss model for the urban macro scenario, as presented in 3GPP TR 38.901, the path loss $h_0$, at the reference distance of the LoS link, is given by
\begin{align}
\centering
h_0 &= 28+22{\rm log_{10}}(d_0)+20{\rm log_{10}}(f),
\end{align} 
where $d_0$ is a reference distance, and $f$ is the sub-carrier frequency. Based on $h_0$ and the proposed channel model in Section \ref{sec2a}, the path loss of $\mathbf{G}$, $\mathbf{H}$, and $\mathbf{r}$ can be estimated, respectively.

\subsection{RIS Elements Allocation and Phase Shifts Configuration}\label{sec3b}
At state S3, the RIS controller optimizes the transmission strategy to maximize the overall throughput of the aerial-terrestrial communication system. Therefore, the maximized overall throughput performance can be written as
\begin{align}\label{Thr}
\centering
{\cal{S}}_{overall} &= \left(1-\frac{T_N}{T_F}\right)R_{overall}^*,
\end{align} 
where $T_N=\sum_{i=1}^{3}t_{si}$ and $T_F=T_N+T_C$ represent the time length of one NP and one frame, respectively, $T_C$ denotes the time length of one CP, and $R_{overall}^*$ is the maximized system capacity which can be obtained by optimizing the RIS group occupation vector ($\mathcal{U}$) and the RIS phase shift matrix ($\mathbf{\Psi}$).

\begin{figure*}[t]	
\centering{\includegraphics[height=2.6in,width=6in]{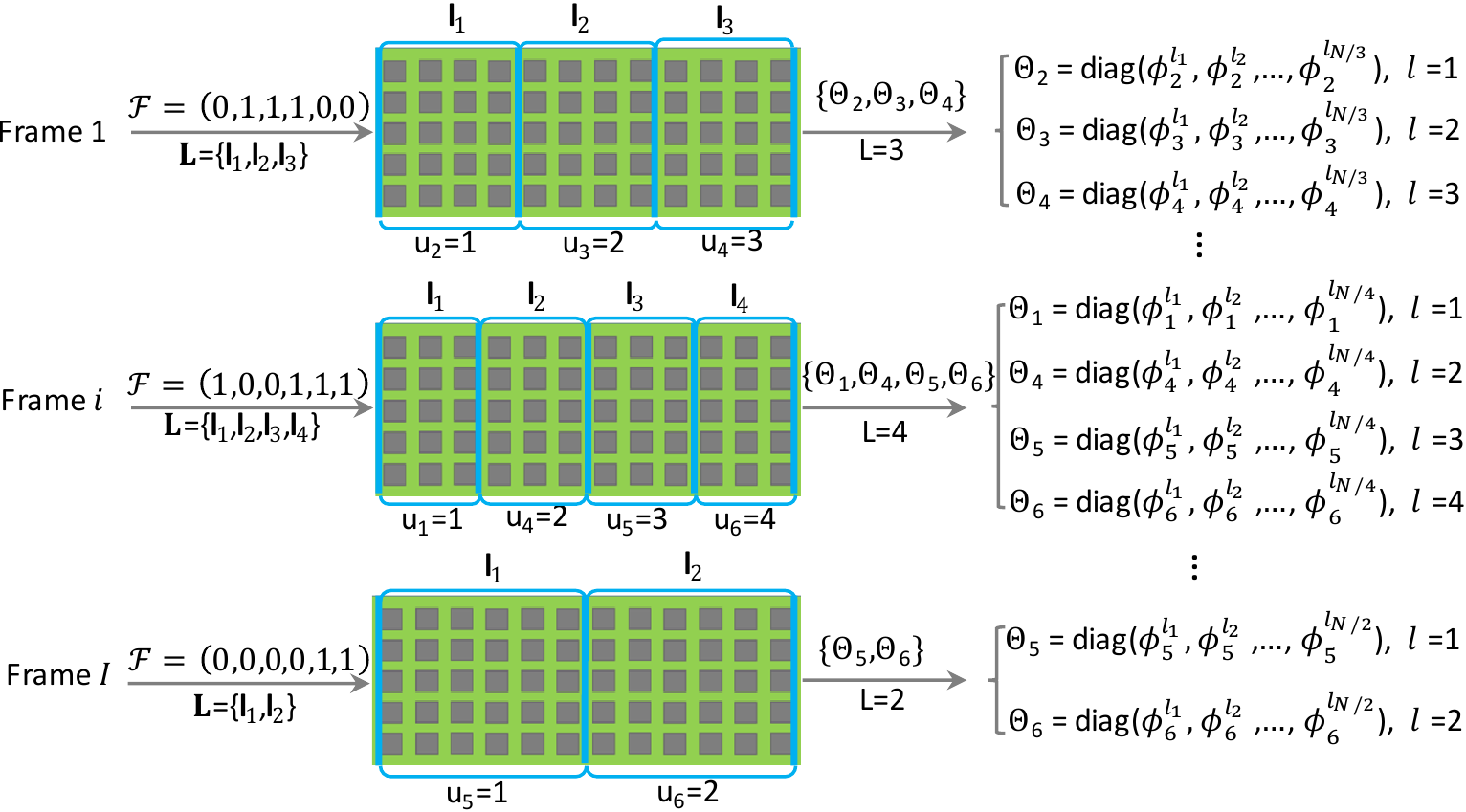}}
	\caption{An illustration of RIS elements allocation and RIS phase shifts configuration, where $K=6$}
	\label{RA}
\end{figure*}

An example related to the RIS elements allocation and RIS phase shifts configuration is shown in Fig. \ref{RA}, when $\mathcal{F}=\{f(u_1),f(u_2),\ldots,f(u_K)\}$ is optimized in one frame, all non-zero values in $\mathcal{F}$ are constructed as a new set that indicates the RIS group partition, denoted by $ \mathbf{L}$. That is to say, $L$ RIS groups are allocated to the $L$ UAV-user pairs with $f(u_k)=1$, in ascending order. For example, the RIS elements allocation in different frames with $K=6$ can be expressed as 
\begin{align}
\centering
\left\{\begin{array}{l}
\!\text{Frame\ } \!1\!: \mathcal{F}\!=\!\{\overbrace{0,1,1,1,0,0}^{L=3}\}\!\Rightarrow\mathbf{L}\! =\! \{\overbrace{\mathbf{l}_1, \mathbf{l}_2, \mathbf{l}_3}^{u_2, u_3, u_4}\}, \\
 \ \ \ldots \ \ \ \ \ \ \ \ \ \ \ \\
\!\text{Frame\ } \!i\!: \mathcal{F}\!=\!\{\overbrace{1,0,0,1,1,1}^{L=4}\}\!\Rightarrow\mathbf{L}\! = \!\{\overbrace{\mathbf{l}_1, \mathbf{l}_2, \mathbf{l}_3,\mathbf{l}_4}^{u_1,u_4,u_5,u_6}\}, \\
 \ \ \ldots \ \ \ \ \ \ \ \ \ \ \ \\
\!\text{Frame\ } \!I\!: \mathcal{F}\!=\!\{\overbrace{0,0,0,0,1,1}^{L=2}\}\!\Rightarrow\mathbf{L}\! = \!\{\overbrace{\mathbf{l}_1, \mathbf{l}_2}^{u_5,u_6}\}. \\
\end{array} \right.
\end{align} 
Specifically, in frame 1, $\mathcal{F}=\{0,1,1,1,0,0\}$ indicates that the RIS is divided into three equal RIS groups, $\mathbf{l}_1, \mathbf{l}_2, \mathbf{l}_3$, which are assigned to the UAV-user pair $2$, $3$, and $4$, respectively. In frame $i$, $\mathcal{F}=\{1,0,0,1,1,1\}$ indicates that the RIS is divided into four equal RIS groups, $\mathbf{l}_1, \mathbf{l}_2, \mathbf{l}_3, \mathbf{l}_4$, which are assigned to the UAV-user pair $1$, $4$, $5$, and $6$, respectively. In frame $I$, $\mathcal{F}=\{0,0,0,0,1,1\}$ indicates that the RIS is divided into two equal RIS groups, $\mathbf{l}_1, \mathbf{l}_2$, which are assigned to the UAV-user pair $5$ and $6$, respectively.

\subsection{Data Communications}\label{sec3c}
Once the NP finishes, the RIS is controlled in the CP to assist the data communications of the UAV-user pairs. As is shown in Fig. \ref{P} (a), if $u_{k}=l$, it means that the $k$-th UAV-user pair communicates via the allocated $l$-th RIS group, then the phase shifts of the $l$-th RIS group are configured accordingly. Otherwise, if $u_{k}=0$, it means that the $k$-th UAV-user pair communicates only via the direct link. Note that the number of RIS groups can be calculated by $L=\sum_{k=1}^{K} f(u_{k})$, and the number of elements on each RIS group equals to $N/L$.

\section{Problem Formulation, Analysis, and Solution}\label{sec4}

According to \eqref{Thr}, the overall system throughput is affected by $R_{overall}^*$ and $T_N$. To maximize the overall throughput, we formulate a joint optimization problem that combines the RIS elements allocation and the RIS phase shifts configuration. Then, we analyze and discuss its solution.
 
\subsection{Problem Formulation}\label{sec4a}
Given the transmission strategy $\mathcal{D}_k=\{u_k, \bm{\theta}_k\}$ of the $k$-th UAV-user pair, the total capacity of the UAV-user pairs with assistance of the RIS can be given by  
\begin{align}\label{RIS}
\centering
R_{RIS} &= \sum_{k=1}^{K}f(u_{k})c_kB\log_{2}\left(1+\text{SNR}_k\right).
\end{align} 
The total capacity of the UAV-user pairs that are without the assistance of the RIS can be written by 
\begin{align}
\centering
R_{DL} &= \sum_{k=1}^{K}\left(1-f\left(u_k\right)\right)c_kB\log_{2}\left(1+\text{SNR}_k\right). 
\end{align}

According to \eqref{SNRE1}, we have 
\begin{align}\label{RIS1}
\centering
R_{RIS} &= \sum_{k=1}^{K}f(u_{k})p_k\omega_1B\log_{2}\left(1\!+\!\frac{\vert\!\left(\hslash_k\!+\!\mathbf{h}_k \mathbf{\Theta}_k \mathbf{g}_k\right)\rho\vert^2}{\sigma^2}\right),
%\notag
%&= \sum_{k=1}^{K}f(u_{k})c_kB\log_{2}\left(1+\frac{\vert\left(\hslash_k+p_kNH_k^l\mathbf{\Theta}_k^lG_k^l\right)\rho\vert^2}{\sigma^2}\right).
\end{align}
and 
\begin{align}\label{DL}
\centering
R_{DL} &=\! \sum_{k=1}^{K}\left(1\!-\!f\left(u_k\right)\right)\!\frac{\omega_2B}{K\!-\sum_{k=1}^{K} f(u_{k})}\log_{2}\left(1\!+\!\frac{\vert \hslash_k\rho\vert^2}{\sigma^2}\right). 
\end{align}
Combined \eqref{RIS1} and \eqref{DL}, we can calculate the overall system capacity, denoted by $R_{oveall}$, which is shown in \eqref{Ga}. 
\begin{figure*}
\begin{align}\label{Ga}
\centering
\!R_{oveall} &\! =\! B\sum_{k=1}^{K}  \!\left(\!f\!\left(u_{k}\!\right)\!p_k\omega_1\!\log_{2}\!\left(1\!+\!\frac{\vert\left(\hslash_k\!+\!\mathbf{h}_k \mathbf{\Theta}_k \mathbf{g}_k\right)\rho\vert^2}{\sigma^2}\right)\!+\!\left(1\!-\!f\left(\!u_{k}\right)\right)\!\frac{\!\omega_2}{\!K-\!\sum_{k=1}^{K} \!f(u_{k})}\!\log_{2}\!\left(1\!+\!\frac{\vert \hslash_k\rho\vert^2}{\sigma^2}\right)\right). 
\end{align} 
	\noindent\rule{\textwidth}{.5pt}
\end{figure*}

To obtain the optimal overall system capacity, $R_{overall}^*$, as shown in \eqref{Thr}, we formulate a joint optimization problem as
\begin{align}\label{p1}
%\begin{split}
\textbf{P1}\;\;\; & \text{(Original problem)}:\ \ \ \ \ \ \ \ \ \ \ \ \ \ \ \ \ \ \ \ \ \ \ \ \ \ \ \ \ \ \ \ \ \ \\
\notag
&\mathop {\max}\limits_{\mathcal{U}, \mathbf{\Psi}}\ \ R_{oveall} \\ 
\notag
\bf{s.t.}\ \
&\text{C1:}\ \ u_k\in\{0,l\}, \ \ \ \ \ \ \ \ \ \ \ \ \ \ \ \forall k \in \mathcal{K}, l\in \mathcal{L}\\
\notag
&\text{C2:}\ \  u_k \neq u_{k'}, \ \ \ \ \ \ \ \ \ \ \ \ \ \ \ \ \ \exists u_k\neq 0, u_k'\neq 0, \forall k, k'\in \mathcal{K},\\
\notag
&\text{C3:}\ \ f(u_{k})\in\{0,1\}, \ \ \ \ \ \ \ \ \ \ \ \forall k \in \mathcal{K},\\
\notag
&\text{C4:}\ \ \sum_{k=1}^{K} f(u_{k}) = L, \ \ \ \ \ \ \ \ \ \ \ \forall k \in \mathcal{K},\\
\notag
&\text{C5:}\ \ L\le L_{max}, \ \ \ \ \ \ \ \ \ \ \ \ \\
\notag
&\text{C6:}\ \ P_o \le P_{max},\ \ \ \ \ \ \ \ \ \ \ \ \ \ \ \ \forall k \in \mathcal{K},\\
\notag
&\text{C7:}\ \ p_k=\{0, \frac{1}{\sum_{k=1}^{K} f(u_{k})}\}, \ \ \forall k \in \mathcal{K}, \\
\notag
&\text{C8:}\ \ \vert\ \phi_k^{l_n} \vert=1, \ \ \ \ \ \ \ \ \ \ \ \ \ \ \ \ \ \ \forall k \in \mathcal{K}, \ \forall l \in \mathcal{L}, \\
\notag
&\text{C9:}\ \ \theta_k^{l_n}\in [0,2\pi),\ \ \ \ \ \ \ \ \ \ \ \ \ \ \ \forall k \in \mathcal{K}, \ \forall l \in \mathcal{L}.
%\end{split}
\end{align}
The constraints in problem \textbf{P1} are detailed as follows: C1 is the RIS group allocation of each UAV-user pair. C2 limits one RIS group to serve one UAV-user pair. C3 denotes whether or not to allocate RIS elements. C4 and C5 constrain the number of RIS groups, where $L_{max}$ is the maximum number of RIS groups. C6 constrains the total power consumption, where $P_{max}$ is the maximization system power consumption. C7 constraints the RIS element occupation ratio of each UAV-user pair. C8 and C9 constrain the reflection coefficients of the RIS.

\subsection{Problem Analysis}\label{sec4b}
For problem \textbf{P1} in \eqref{p1}, we analyze the maximized overall system capacity in terms of three cases, which are presented as follows.    
\subsubsection{$K$ UAV-user pairs communicate via the RIS} When the number of UAV-user pairs is small, each UAV-user pair can communicate with the assistance of the RIS, which means that the RIS elements are enough to support and improve the communications of all UAV-user pairs. In this case, we give \textbf{Observation 1} to show the maximized overall system capacity.  
\begin{observation}
For $K$ UAV-user pairs, the maximized overall system capacity with the assistance of the RIS is calculated as
\begin{align}\label{Ca}
\centering
\!R^*_{RIS}\!&= \!\mathop{\max}\limits_{\mathcal{U},\mathbf{\Psi}}\ \sum_{k=1}^{K}\!f(u_k)p_k\!\omega_1B\!\log_{2}\!\left(1\!+\!\frac{\left|\left(\!\hslash_k\!+\!\mathbf{h}_k \mathbf{\Theta}_k \mathbf{g}_k\right)\rho\right|^2}{\sigma^2}\right)\\
\notag
&= \!\frac{B}{K}\sum_{k=1}^{K}\!\log_{2}\left(1\!+\!\frac{\rho^2\left(\vert \hslash_k \vert\!+\!\sum_{n=1}^{\frac{N}{K}}\!\vert h_k^{l_n} \vert \vert g_k^{l_n}\vert\right)^2}{\sigma^2}\right). 
\end{align} 
\end{observation}
\begin{proof}
Since $K$ UAV-user pairs communicate simultaneously via the RIS, we have $f(u_k)=1, \forall k \in \mathcal{K}$, $L=\sum_{k=1}^{K}f(u_k)=K$, $c_k=\frac{\omega_1B}{K}$, $\omega_1=1$, and $\omega_2=0$, respectively. According to \eqref{RIS1}, for any given $\mathbf{\Theta}_k$, the overall maximization capacity can be expressed as in \eqref{Ca}, which can be calculated by the given channel information of $\mathbf{h}_k$ and $\mathbf{g}_k$. Note that $\mathbf{h}_k\mathbf{\Theta}_k\mathbf{g}_k=\sum_{n=1}^{\frac{N}{K}}e^{j\theta_k^{{l_n}}}\vert h_k^{l_n} \vert \vert g_k^{l_n}\vert$, where $h_k^{l_n}\in \mathbf {h_k}$, $g_k^{l_n}\in \mathbf {g_k}$, $\forall l_n\in \mathbf{l}_l, \forall n \in (1,N/K)$. The maximized overall system capacity can be obtained when the optimal phase shifts are set as $\theta_k^{{l_n}^*} = \arg(\hslash_k)-\arg(h_{k}^{l_n})-\arg(g_{k}^{l_n})$ to align the direct path of the $k$-th UAV-user pair, i.e., enable each term in the sum with the same phase shift as $\hslash_k$.    
\end{proof} 

\subsubsection{$K$ UAV-user pairs communicate without the RIS } When all UAV-user pairs communicate without the RIS, each UAV-user pair communicates relying on the direct path. With this case, we give \textbf{Observation 2} to show the maximized overall system capacity.
\begin{observation}
For $K$ UAV-user pairs, the maximized overall system capacity without the assistance of the RIS is given by
\begin{align}\label{Cda}
\centering
\!R^*_{DL} \!&= \!\mathop {\max}\limits_{\mathcal{U}}\ \! \sum_{k=1}^{K}\!\left(1\!-\!f\left(u_k\right)\!\right)\!\frac{\!\omega_2B}{K\!-\!\sum_{k=1}^{K} \!f(u_{k})}\!\log_{2}\!\left(1\!+\!\frac{\vert \hslash_k\rho\vert^2}{\sigma^2}\right)\\
\notag
&= \frac{B}{K}\sum_{k=1}^{K} \log_{2}\left(1+\frac{\left|\hslash_k\rho\right|^2}{\sigma^2}\right). 
\end{align} 
\begin{proof}
Since $K$ UAV-user pairs communicate without the assistance of the RIS, we have $f(u_k)=0, \forall k \in \mathcal{K}$, $L=\sum_{k=1}^{K}f(u_k)=0$, $\omega_1=0$, and $\omega_2=1$, respectively. According to \eqref{DL}, the maximized overall system capacity can be expressed as in \eqref{Cda}, which is determined by the channel information $\hslash_k$.     
\end{proof} 
\end{observation}

\subsubsection{$L$ UAV-user pairs communicate via the RIS while $K-L$ UAV-user pairs communicate without the assistance of the RIS} When the number of UAV-user pairs is larger than the maximum number of RIS groups (i.e., $L_{max} < K$), the limited number of RIS elements cannot meet all the requirements of $K$ UAV-user pairs. Therefore, we suppose that $L$ UAV-user pairs are allocated RIS elements to support their data communications, and $K-L$ UAV-user pairs communicate relying on the direct path without the assistance of the RIS. Here, the number of RIS group, $L$, is given by $L=\sum_{k=1}^{K}f(u_k)$. In this case, considering the RIS elements allocation and RIS phase shifts configuration, we give \textbf{Observation 3} to show the maximized overall system capacity. 

\begin{observation}
For $K$ UAV-user pairs, the maximized overall system capacity with limited number of the RIS elements is written as
\begin{align}\label{Ga1}
\centering
R^*_{overall} \!&= \!\mathop {\max} \limits_{\mathcal{U}, \mathbf{\Psi}}\ R_{overall}\\
\notag
\!&=\!\mathop {\max} \limits_{\mathcal{U}, \mathbf{\Psi}} \left(\!\sum_{k=1}^{K}\frac{\!\omega_1B}{L}f(u_k)\!\log_{2}\!\left(1\!+\!\frac{\vert\left(\hslash_k\!+\!\mathbf{h}_k \mathbf{\Theta}_k \mathbf{g}_k\!\right)\!\rho\vert^2}{\sigma^2}\!\right)+\sum_{k=1}^{K}\frac{\omega_2B}{K-L}(1-f(u_k))\log_{2}\left(1+\frac{\vert \hslash_k\rho\vert^2}{\sigma^2}\right)\right). \notag 
\end{align}  
\begin{proof}
Since $L$ UAV-user pairs communicate via the RIS while $K-L$ UAV-user pairs communicate without the assistance of the RIS, we have $f(u_k)\in\{0,1\}$, $L=\sum_{k=1}^{K}f(u_k)\in(0,K)$, and $p_k=1/L$, respectively. According to \eqref{Ga}, the maximized overall system capacity can be expressed as in \eqref{Ga1}. The solution of \eqref{Ga1} can be achieved by optimizing the transmission strategy, including RIS elements allocation and RIS phase shifts configuration of each UAV-user pair, as shown in the following Section \ref{sec4c}.  
\end{proof}
\end{observation}

\subsection{Problem Solution}\label{sec4c}
Based on the aforementioned analysis, the original problem \textbf{P1} in \eqref{p1} can be rewritten as  
\begin{align}\label{p2}
\textbf{P2}:\;\;&\mathop {\max} \limits_{\mathcal{U}, \mathbf{\Psi}}\left(\!\sum_{k=1}^{K}\frac{\!\omega_1B}{L}f(u_k)\!\log_{2}\!\left(1\!+\!\frac{\vert\left(\hslash_k\!+\!\mathbf{h}_k \mathbf{\Theta}_k \mathbf{g}_k\!\right)\!\rho\vert^2}{\sigma^2}\!\right)+\sum_{k=1}^{K}\frac{\omega_2B}{K-L}(1-f(u_k))\log_{2}\left(1+\frac{\vert \hslash_k\rho\vert^2}{\sigma^2}\right)\right)
\notag 
\\
\bf{s.t.}\ \
&\text{C1-C6,}\;\;\text{C8-C9.}
\end{align}
The joint optimization problem \textbf{P2} is an MINLP problem. To solve this MINLP problem, \textbf{P2} can be further decomposed into two sub-problems with separated objective function and constraints. On this basis, the transmission strategy problem \textbf{P2} can be solved by alternately optimizing the RIS elements allocation and the RIS phase shifts configuration.

\subsubsection{RIS elements allocation optimization}\label{sec4c1} Given the RIS phase shifts configuration, $\mathbf{\Psi}$, problem \textbf{P2} in \eqref{p2} is transformed into
\begin{align}\label{p21}
\textbf{P2-1}:\;\;&\mathop {\max} \limits_{\mathcal{U}}\left(\!\sum_{k=1}^{K}\frac{\!\omega_1B}{L}f(u_k)\!\log_{2}\!\left(1\!+\!\frac{\vert\left(\hslash_k\!+\!\mathbf{h}_k \mathbf{\Theta}_k \mathbf{g}_k\!\right)\!\rho\vert^2}{\sigma^2}\!\right)+\sum_{k=1}^{K}\frac{\omega_2B}{K-L}(1-f(u_k))\log_{2}\left(1+\frac{\vert \hslash_k\rho\vert^2}{\sigma^2}\right)\right) 
\notag
\\
\bf{s.t.}\ \ \ \ 
&\text{C1-C6.}\;\;
\end{align}

Referring to the proposed RIS group partition in Section \ref{sec3b}, the optimal $\mathcal{U}^*$ will be obtained as the RIS controller solves the optimal $\mathcal{F}^*$. Therefore, Problem \textbf{P2-1} in \eqref{p21} can be transformed into an RIS-assisted transmission decision problem, which is illustrated as
\begin{align}\label{p21a}
\textbf{P2-1a}:\;&\mathop {\max} \limits_{\mathcal{F}}\left(\!\sum_{k=1}^{K}\frac{\!\omega_1B}{L}f(u_k)\!\log_{2}\!\left(1\!+\!\frac{\vert\left(\hslash_k\!+\!\mathbf{h}_k \mathbf{\Theta}_k \mathbf{g}_k\!\right)\!\rho\vert^2}{\sigma^2}\!\right)+\sum_{k=1}^{K}\frac{\omega_2B}{K-L}(1-f(u_k))\log_{2}\left(1+\frac{\vert \hslash_k\rho\vert^2}{\sigma^2}\right)\right)\notag
\\ 
\bf{s.t.}\ \ \ \ \
&\text{C3-C6.}\;\;
\end{align}
When the number of UAV-user pairs is small, we can obtain a high-quality global optimal solution of problem \textbf{P2-1a} by using the exhaustive method. Since $f(u_k)$ equals to ``$0$" or ``$1$", the scale of $\mathcal{F}=\{f(u_1),f(u_2),\ldots,f(u_K)\}$ is $2^K$. Therefore, we can solve the optimal $\mathcal{F}^*$ to maximize the system capacity by searching each case of $\mathcal{F}$. According to the optimal $\mathcal{F}^*$ and the proposed scheme of RIS group partition, the optimal $\mathcal{U}^*$ can be solved.    

\subsubsection{RIS phase shifts configuration optimization}\label{sec4c2} Given the RIS elements allocation, $\mathcal{U}$, problem \textbf{P2} in \eqref{p2} can be transformed as
\begin{align}\label{p22}
\textbf{P2-2}:\;&\mathop {\max} \limits_{\mathbf{\Psi}}\left(\!\sum_{k=1}^{K}\frac{\!\omega_1B}{L}f(u_k)\!\log_{2}\!\left(1\!+\!\frac{\vert\left(\hslash_k\!+\!\mathbf{h}_k \mathbf{\Theta}_k\mathbf{g}_k\!\right)\!\rho\vert^2}{\sigma^2}\!\right)+\!\sum_{k=1}^{K}\!\frac{\omega_2B}{K\!-\!L}(1\!-\!f(u_k))\log_{2}\!\left(1\!+\!\frac{\vert \hslash_k\rho\vert^2}{\sigma^2}\!\right)\!\right)\!\notag 
\\
\bf{s.t.}\ \ \ \
&\text{C8-C9.}
\end{align}

When $\mathcal{U}$ is given, the optimal $\mathbf{\Psi}$ is the one that maximizes the channel gain of each UAV-user pair with the assistance of the RIS, i.e., solve $\bm{\theta}_k$ to maximize $\vert\left(\hslash_k\!+\!\mathbf{h}_k \mathbf{\Theta}_k \mathbf{g}_k\!\right)\vert^2$. On this basis, the optimal problem of $\mathbf{\Psi}$ can be transformed into the following problem:
\begin{align}\label{p22a}
\textbf{P2-2a}:\;\;&\mathop {\max} \limits_{\bm{\theta}_k}\ \vert \hslash_k +\mathbf{h}_k \mathbf{\Theta}_k \mathbf{g}_k \vert^2\ \ \ \ \ \ \ \ \ \ \ \ \ \ \ \ \ \ \ \ \ \ \ \ \ \ \ 
\notag 
\\
\bf{s.t.}\ \ \ \ \
&\text{C8-C9.}
\end{align}

Due to the non-convexity of problem \textbf{P2-2a} in \eqref{p22a}, in the following, we solve this problem in detail, as highlighted in\textbf{ Remark \ref{r1}}. 

\begin{remark}\label{r1}
The optimal solution of problem in \eqref{p22a} is $\theta_k^{{l_n}^*} = \arg(\hslash_k)-\arg({h}_{k}^{l_n})-\arg({g}_{k}^{l_n})$, where $l_n \in \mathbf{l}_l$. With considering a same reflection coefficient for all elements in one RIS group (i.e., $\theta_k^l=\theta_k^{{l_n}},\forall l_n \in \mathbf{l}_l$), we have $\theta_k^{{l}^*} = \arg(\hslash_k)-\arg({h}_{k}^{l})-\arg({g}_{k}^{l})$. 
\end{remark}

\begin{proof}
For the RIS phase shifts configuration optimization of problem \textbf{P2-2a}, we have the following equality
\begin{align}\label{ph}
\centering
\vert \hslash_k +\mathbf{h}_k \mathbf{\Theta}_k \mathbf{g}_k \vert^2 \le \vert \hslash_k \vert^2+\vert \mathbf{h}_k \mathbf{\Theta}_k \mathbf{g}_k\vert^2 . 
\end{align} 
The equality in \eqref{ph} holds only when RIS reflection coefficients satisfies $\arg(\hslash_k) \triangleq \arg(\mathbf{h}_k \mathbf{\Theta}_k \mathbf{g}_k)$. To optimize $\theta_k^{{l_n}}, l_n \!\in\!\mathbf{l}_l$, let $\mathbf{h}_k \mathbf{\Theta}_k \mathbf{g}_k\!=\!\mathbf{v}_k\mathbf{\Phi}_k$, where $\mathbf{v}_k\!=\![v_k^{l_1},\ldots, v_k^{l_{N/L}}]\!\in\!\mathbb{C}^{1\times(N/L)}$, $v_k^{l_n}\!=\!e^{j\theta_k^{{l_n}}}, \forall l_n \in \mathbf{l}_l$, and $\mathbf{\Phi_k}\!=\!\text{diag}(\mathbf{h}_k)\mathbf{g}_k$. problem \textbf{P2-2a} can be simplified as 
\begin{align}\label{p22b}
\textbf{P2-2b}:\;\;&\mathop {\max} \limits_{\mathbf{v}_k}\ \vert \mathbf{v}_k \mathbf{\Phi}_k \vert^2\ \ \ \ \ \ \ \ \ \ \ \ \ \ \ \ \ \ \ \ \ \ \ \ \ \ \ \\
\notag 
\bf{s.t.}\ \
&\text{C10:}\ \  \vert {v}_k^{l_n} \vert^2=1, \ \ \forall l_n \in \mathbf{l}_l.
\end{align}
It is observed that the optimal RIS phase shifts on the $l$-th RIS group for the $k$-th UAV-user pair can be obtained by given $\mathbf{v}_k^*=e^{j({\arg(\hslash_k)-\arg(\text{diag}(\mathbf{h}_k)\mathbf{g}_k))}}$, and then we have  
\begin{align}\label{p22bb}
\theta_k^{{l_n}^*} = \arg(\hslash_k)-\arg({h}_{k}^{l_n})-\arg({g}_{k}^{l_n}),
\end{align}
where $h_k^{l_n}\in \mathbf {h_k}$, $g_k^{l_n}\in \mathbf {g_k}$, $\forall l_n\in \mathbf{l}_l, \forall n \in (1,N/L)$. It is noting that if a same reflection coefficient is simplified for all elements in one RIS group, the optimal reflection coefficient of the $l$-th RIS group is $\theta_k^{{l}^*} = \arg(\hslash_k)-\arg(h_{k}^{l})-\arg(g_{k}^{l})$, where $h_k^{l_n} = {h_k^l}$, $g_k^{l_n} = {g_k^l}$, $\forall l_n\in \mathbf{l}_l$. In other words, each element on the same RIS group configures its phase shift to the same reflection coefficient to assist one UAV-user pair's communication.  
\end{proof}

We analyze the complexity of solving problem \textbf{P2}, as illustrated in \textbf{Remark \ref{r2}}. 
\begin{remark}\label{r2}
The complexity of solving the problem in \eqref{p2} is dominated by calculating $\cal U$ in \eqref{p21} and $\mathbf{\Psi}$ in \eqref{p22}, the complexity of each is ${\cal O} (2^K)$ and ${\cal O} (N)$, respectively. Hence, the total complexity of solving the problem in \eqref{p2} is ${\cal O} (2^KN\ell)$, where $\ell$ is the number of iterations. We can see that the complexity of Problem \textbf{P2} solution increases with the number of RIS elements and the number of UAV-user pairs. Especially, when all elements of one RIS group are assumed to maintain the same reflection coefficient, the complexity can be decreased to ${\cal O} (2^KL\ell)$.
\end{remark}

The real-time RIS elements allocation and phase shifts configuration are costly to implement in practice because manufacturing such high-precision elements requires substantial computation and expensive hardware. However, the conventional optimization procedure must be executed repeatedly as the parameters change, which will incur high computational complexity due to numerical iterations, and its solutions are often suboptimal and do not scale well. An artificial intelligence-based method that moves the complexity of online computation to offline training is a feasibility scheme to improve the efficiency of obtaining the problem solution, thereby improving RIS utilization with high accuracy and low cost. As such, the multi-task learning method is explored to address the real-time optimization problem efficiently and improve RIS utilization.

\begin{figure}[t]	
\centering{\includegraphics[height=2.2in,width=2.5in]{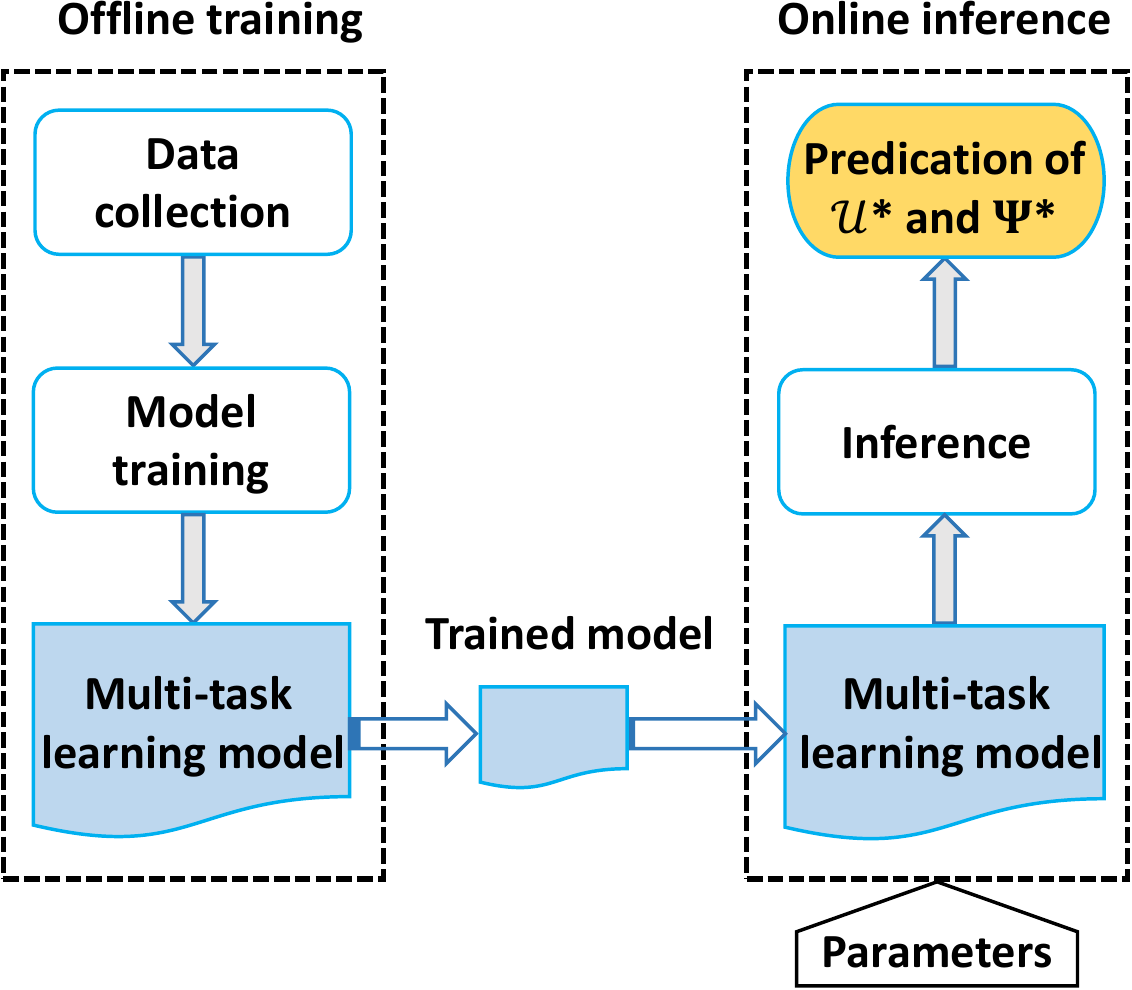}}
	\caption{The proposed multi-task learning model.}
	\label{LM}
\end{figure}

\section{Multi-task Learning for Transmission Strategy on The RIS}\label{sec5}
In this section, we present a multi-task learning model instead of the conventional iterative method to speed up the decision process and decrease the computation complexity at the RIS controller.

\subsection{Multi-task Learning Model}\label{sec5a}
As is shown in Fig. \ref{LM}, deep neural networks can be designed and trained offline to obtain a multi-task learning model. By feeding the parameters (such as the number of UAV-user pairs, RIS size, and channel conditions) into the trained multi-task learning model, the optimal transmission strategy of each UAV-user pair (including the optimal RIS elements allocation and the optimal RIS phase shifts configuration) can be predicted at the output. Since moving the complexity of online computation to offline training, the complexity of solving problem \textbf{P2} is incurred by online inference. Therefore, the solution to \textbf{P2} can be obtained efficiently by performing feedforward calculation without iterations \cite{yang2020computation}, the complexity can be significantly decreased to ${\cal O} (1)$.

\subsection{Problem Transformation}\label{sec5b}
Owing to the interaction between two independent output vectors (i.e., $\mathcal{U}^*$ and $\mathbf{\Psi}^*$) of problem $\mathbf{P2}$, we consider the prediction of $\mathcal{U}^*$ and $\mathbf{\Psi}^*$ as two individual machine learning tasks. Compared to training the models separately, the learning efficiency and prediction accuracy can be improved with the multi-task learning structure \cite{SR}. On this basis, we formulate $\mathbf{P2}$ as a multi-task learning problem in Fig. \ref{AI}. We suppose that there exist total $A$ learning tasks $\left \{ {\Omega}_i \right\}{_{i=1}^{A}}$ that are interactive with each other, where ${A}=2$ in our proposed multi-task learning model. Each learning task ${\Omega}_i$ is performed with a training dataset ${\Upsilon}_i$, which consists of $J$ training samples. Then, we have 
\begin{equation}\label{e23}
{\Upsilon}_i = \left \{ \mathbf{X}_{j}^{(i)},\mathbf{Y}_{j}^{(i)} \right \}, \ j=1, \ldots , J,
\end{equation}
where $\mathbf{X}_{j}^{(i)}$ is the $j$-th training instance in ${\Omega}_i$, $\mathbf{Y}_{j}^{(i)}$ represents its label. 

For the output, denote $y_{j}^{(i)}$ as the $j$-th corresponding output from $\mathbf{Y}_{j}^{(i)}$. When $y_{j}^{(i)}$ is in a discrete space, e.g., $y_{j}^{(i)}\in \left \{ 0, 1 \right \}$ for the binary RIS-assisted transmission decision, the corresponding task can be regarded as a \textit{classification problem} aiming to predict RIS allocation for a given set of input parameters. If $y_{j}^{(i)}$ is continuous, e.g., $y_{j}^{(i)}\in \mathbb{R}$ is for the phase shift of the RIS element, the corresponding task can be transformed into a \textit{regression problem}, which aims to predict a numeric value. For simplicity, in the proposed regression model, we predict $\mathbf{\Psi}^*/2\pi$ instead of predicting $\mathbf{\Psi}^*$, which makes the prediction of the regression model becomes a normalized ratio of RIS phase shifts.

\begin{figure*}[t]	
\small
\centering{\includegraphics[width=1\textwidth]{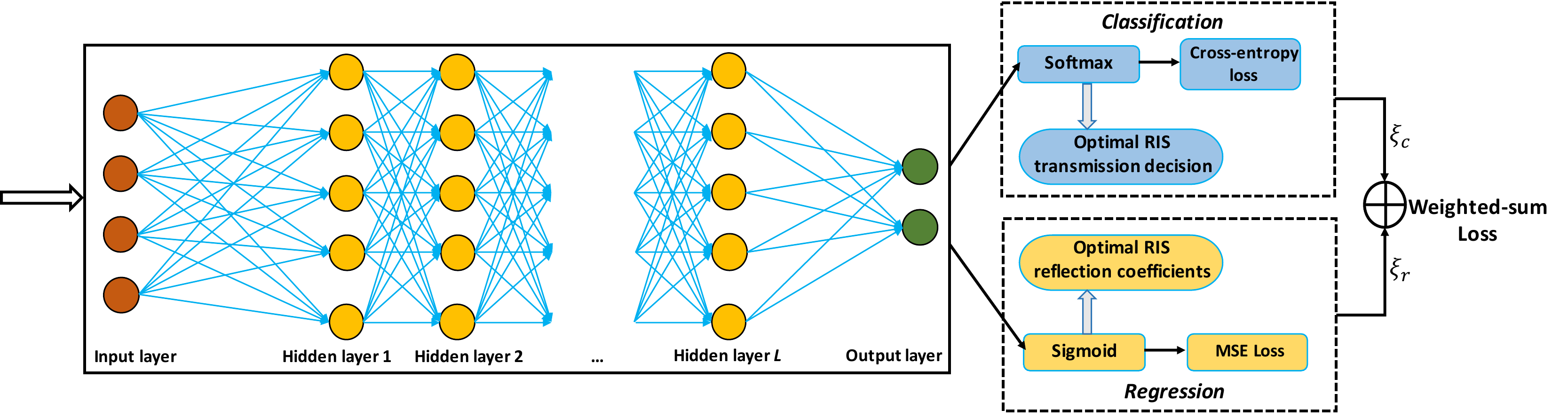}}
	\caption{The framework of the trained multi-task learning model.}
	\label{AI}
\end{figure*}

\renewcommand{\algorithmicrequire}{\textbf{Initialization:}} 
\renewcommand{\algorithmicensure}{\textbf{Output:}} 
\begin{algorithm}[t]         
\caption{Dataset Collection}             
\label{alg2}                  
\begin{algorithmic}[1]             
\REQUIRE  Iteration index $i=0$ and dataset $\Upsilon=\varnothing$;
\ENSURE Dataset $\Upsilon$ \\  %算法的迭代：Iteration
\WHILE    {$i< {\rm dataset\ size} \ $}
\STATE $i\leftarrow i+1$;
\STATE Generate input parameters set ($\mathbf{X}_i$) for all UAV-user pairs;
\STATE Solve $\mathbf{P2}$ by using exhaustive searching and record the optimal solution as $\mathbf{Y}_{i} = ( {\mathcal U_i}^*$, ${\mathbf \Psi}_i^*/2\pi)$;
\STATE Add the $i$-th input/output pair $\left \{ {\mathbf{X}}_{i},\mathbf{Y}_{i} \right \}$ to $\Upsilon$.
\ENDWHILE             
\end{algorithmic}
\end{algorithm}

\renewcommand{\algorithmicrequire}{\textbf{Initialization:}}  
\renewcommand{\algorithmicensure}{\textbf{Output:}} 
\begin{algorithm}[t]       
\caption{Offline Training and Online Inference}
\label{alg3}             
\textit{\textbf{1) Offline Training}}\;                  
\begin{algorithmic}[1]           
\STATE Build dataset with \textit{Algorithm 1};
\STATE Train the classifier with loss function $\iota_c$ given in (\ref{lc});
\STATE Train the regressor with loss function $\iota_r$ given in (\ref{lr});
\STATE Achieve the weighted-sum loss function $\iota$ based (\ref{l});
\STATE Tune the weights of each layer using backpropagation until $\iota$ is minimized.
\end{algorithmic}
\textit{\textbf{2) Online Inference}\;}                 
\begin{algorithmic}[1]           
\REQUIRE $K$, $N$, $B$, $\omega_1$, $\omega_2$, $\rho^2$, $P_{max}$, $P_{R}$, $P_{k, U}$, $I$, $L_{max}$, and $i=0$;
\ENSURE Optimal transmission strategy $\mathbf{D_i}^*$; 
\WHILE {$ i< I \ $}
\STATE $i\leftarrow i+1$;
\STATE Input UAVs' parameters to the pre-trained multi-task learning model at the RIS controller;
\STATE RIS controller predicts the transmission decision $\mathcal{U}_i^*$ and RIS phase shifts $\mathbf{\Psi}_i^*/2\pi$;
\STATE RIS controller obtains the optimal transmission strategy $\mathbf{D}_i^*=\{\mathcal{U}_i^*,{\mathbf{\Psi}_i}^*\}$;
\STATE RIS controller manages RIS elements allocation and RIS reflection coefficients configuration according to $\mathbf{D}_i^*$;
\ENDWHILE

\end{algorithmic}
\end{algorithm} 

\subsection{Dataset Collection, Offline Training, and Online Inference}\label{sec5c}

The process of training dataset generation and collection are shown in \textbf{Algorithm~\ref{alg2}}. Here, we obtain the dataset using the exhaustive searching to solve the optimization problem \textbf{P2}. Note that although the input parameters during offline training do not include every possible combination of the parameters, they do cover a wide range of parameter settings, and the well-known generalization property of machine learning models will enable the trained deep learning model to produce accurate inference of the solution, even for parameter settings not included in the training samples. Therefore, generalization can be achieved by designing a proper deep learning model with enough labelled data. This indicates that the dataset input to the multi-task learning model is with very good quality, which guarantees the accuracy on solving the optimization problem.
 
The framework of the trained multi-task learning model highlighted in Fig.~\ref{AI}, in which the offline training and the online inference are illustrated in \textbf{Algorithm~\ref{alg3}}, respectively. During the offline training, based on the collected dataset, we use back-propagation to train the multi-task learning model. The online inference process can be divided into one classification problem and one regression problem, both of them are combined to predict $\mathcal{U}$ and $\mathbf{\Psi}$ for all UAV-user pairs.

For the classification problem, the probability of each class is predicted using the Softmax function, i.e., the predicted probability for the $q$-th class is given as
\begin{equation}\label{si}
\varrho_q(z) =\frac{e^{z_q}}{\sum_{q=1}^{Q} e^{z_q}}, \ q=1,\ldots,Q,
\end{equation}
where $Q$ is the total number of classes, $z$ is the output of the last fully connected layer. 

The loss function of the classification, denoted by $\iota_c$, is defined as the cross-entropy. Then, we have
\begin{equation}\label{lc}
{\iota_c}=-\frac{1}{Q}\sum_{q=1}^{Q}\mathbf{Y}_q{\rm{ln}}g(\mathbf{X}_q),
\end{equation}
where ${\mathbf{X}}_q$ is input devices' context, $\mathbf{Y}_q$ denotes the ground truth, and $g(\mathbf{X}_q)$ is the actual output of neurons.

For the regression problem, with Sigmoid function, the numerical ratio value can be mapped between 0 and 1, then the prediction value is given by 
\begin{equation}\label{si1}
\varrho(z) =\frac{1}{1+e^{-z}}.
\end{equation} 

We define $M$ as the number of input samples. By using the mean square error (MSE), the loss function of the regression, denoted by $\iota_r$, is expressed as
\begin{equation}\label{lr} 
{\iota_r}=\frac{1}{M}\sum_{i=1}^{M} \left (\mathbf{Y}_i- g(\mathbf{X}_i) \right )^2.
\end{equation}

The loss function of the proposed multi-task learning model is regarded as the weighted-sum of ${\iota_c}$ and ${\iota_r}$, which is given by
 \begin{equation}\label{l}
\iota = \xi_c {\iota_c}+\xi_r {\iota_r},
\end{equation}
where $\xi_c$ and $\xi_r$ denote the weights of the classification and the regression, respectively.

\begin{figure}[t]	
\centering{\includegraphics[width=0.5\textwidth]{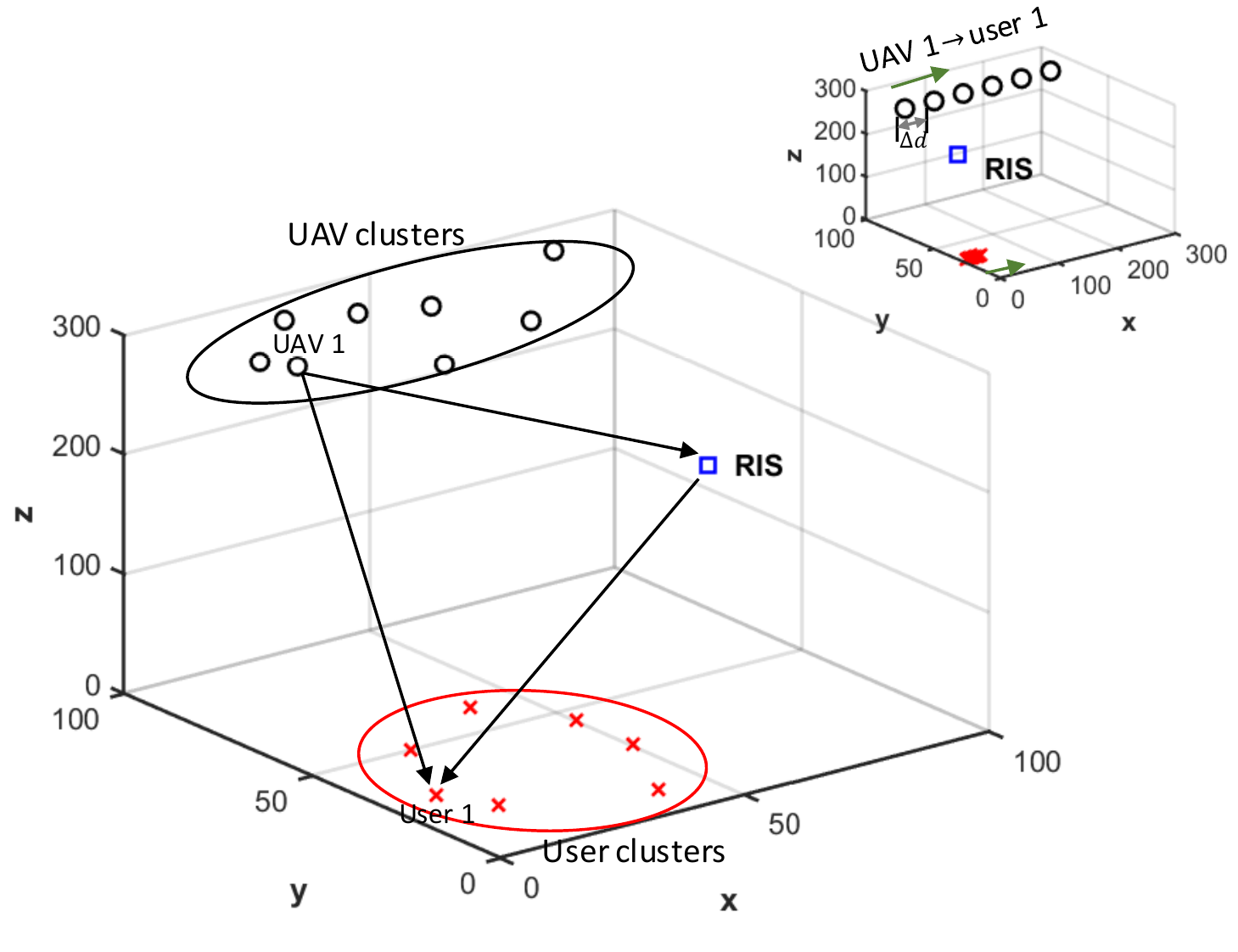}}
	\caption{The simulation scenario.}
	\label{SC}
\end{figure}

\section{Simulation Results}\label{sec6}
We provide the simulation results to validate the effectiveness of RIS-assisted aerial-terrestrial communications. The simulation scenario is presented in Fig. \ref{SC}, where the UAV clusters and user clusters are randomly deployed within a 3D area of 100 m $\ast$ 100 m $\ast$ 300 m, all UAVs and users maintain a certain altitude. Furthermore, we consider the UAV-user pair only move at x-axis along the same direction, the moving speed of UAV and user is 25 $m/s$ and 0.5 $m/s$, respectively. The location of RIS is fixed, and the other main parameters are set as in Table \ref{SP}.

\begin{table}[t] 
		\small 
		\centering
			\renewcommand{\arraystretch}{1.2}
			\captionsetup{font={small}} 
			\caption{\scshape Simulation parameters} 
			\label{SP}
			\footnotesize
			\centering  
			\begin{tabular}{| c | c | c | c |}  
				\shline
				\textbf{Notation} & \textbf{Definition} &\textbf{Notation} & \textbf{Definition}\\
				\shline 
				UAV-user pairs \!($K$\!) & 8 & RIS elements \!($N$\!) & 512\\ 
			    \hline 
			    RIS groups \!($L$\!) & [1  2  4  8]  & Sub-carriers \!($C$\!) & 8\\ 
			    \hline 
			    Noise power \!($\sigma^2$\!) & -94 dBm & UAVs power \!($\rho^2$\!) & 10 mW\\
			    \hline 
				UAV antenna gain & 10 dBi & Bandwidth & 10 MHz \\ 
				\hline 
				User antenna gain & 5 dBi  & Frequency & 5 GHz \\  
			    \hline
			    UAVs altitude & (250, 300) m  & Users altitude & (0, 1) m\\  
			   	\hline
			   	User 1 location &  \![10  30  1] m & Weight \!($\omega_1,\omega_2$) & 0.6, 0.4\\ 
			   	\hline
			   	UAV 1 location & \![20  80  280] m & Weight \!($\xi_1,\xi_2$) & 0.5, 0.5 \\ 
			    \hline 
			    RIS location & \![\!100 75 120] m & Frame length &  1 ms\\ 
			    \shline

			\end{tabular}  
	\end{table}
	
Figure \ref{S1} evaluates the SNR, the transmit power, and the throughput in terms of the number of RIS groups and the number of UAV-user pairs. Specifically, Fig. \ref{S1} (a) shows that the SNR of UAV-User pair 1 (Pair $\#$1 in short) varies as the number of RIS groups ($L$) increases when Pair $\#$1 moves at the x-axis along the same direction. On the one side, we can see that the SNR decreases as the number of RIS groups increases because fewer RIS elements can be utilized in each RIS group. On another side, we also see the SNR varies as the location of Pair $\#$1 (i.e., the moving distance ($\Delta d$) at the x-axis for UAV 1 is (0; 50; 100; 150; 200; 250), and for user 1 is (0; 5; 10; 15; 25; 25), where the unit is meter (m)). It can be seen that the best SNR can be achieved at the location [120 80 280; 20 30 1] of Pair $\#$1, and then followed by [70 80 280; 15 30 1], [170 80 280; 20 30 1], [20 80 280; 10 30 1], and [220 80 280; 30 30 1]. The worst SNR is observed at the location [270 80 280; 35 30 1] of Pair $\#$1. It can be seen that the SNR is affected by the channel gain between Pair $\#$1 and the RIS. In other words, a higher SNR can be achieved when Pair $\#$1 is close to the RIS, and a lower SNR exists when Pair $\#$1 is far away from the RIS. Figure \ref{S1} (b) shows that the transmit power changes as the number of RIS groups increases when Pair $\#$1 moves at the x-axis along the same direction. Compared with the SNR in Fig. \ref{S1} (a), the transmit power displays an opposite changing trend as the number of RIS groups and the location of Pair $\#$1 change. It means that the fewer RIS groups and the closer distance to the RIS can reduce the energy consumption of Pair $\#$1. Figure \ref{S1} (c) shows that the throughput decreases as the number of UAV-user pairs ($K$) and the maximum number of RIS groups ($L_{max}$) increase. Four schemes are presented to evaluate aerial-terrestrial communications, i.e., without the RIS, with the RIS random phase shift, with the alternating algorithm-based RIS, and with the multi-task learning method (MTL)-based RIS. Compared to aerial-terrestrial communications without the RIS and RIS random phase shift, RIS-assisted aerial-terrestrial communications adopted the alternating algorithm and the proposed MTL method show significant improvement (increase from 0.06 Mbps to 28 Mbps) since the optimal RIS configuration can enhance the links performance of UAV-user pairs. Moreover, compared to the alternating algorithm, the MTL method achieves better throughput and also with much lower complexity. Given a fixed $L_{max}$ as the alternating algorithm or the MTL method optimizes the RIS, the system throughput decreases as the number of UAV-user pairs increases because the poor direct link communications need to consume the bandwidth. Besides, given a fixed $K$, the system throughput decreases as the maximum number of RIS groups increases because the less bandwidth and fewer RIS elements can be used for the UAV-user pairs that have the better channel states. Note that the large $L_{max}$ means that the RIS can serve the more UAV-user pairs at a time to improve their links.

\begin{figure*}[t]
\small
      \centering
    	\subfigure[SNR vs. number of RIS groups]{
    		\begin{minipage}[b]{0.315\textwidth}
    			\includegraphics[width=1\textwidth]{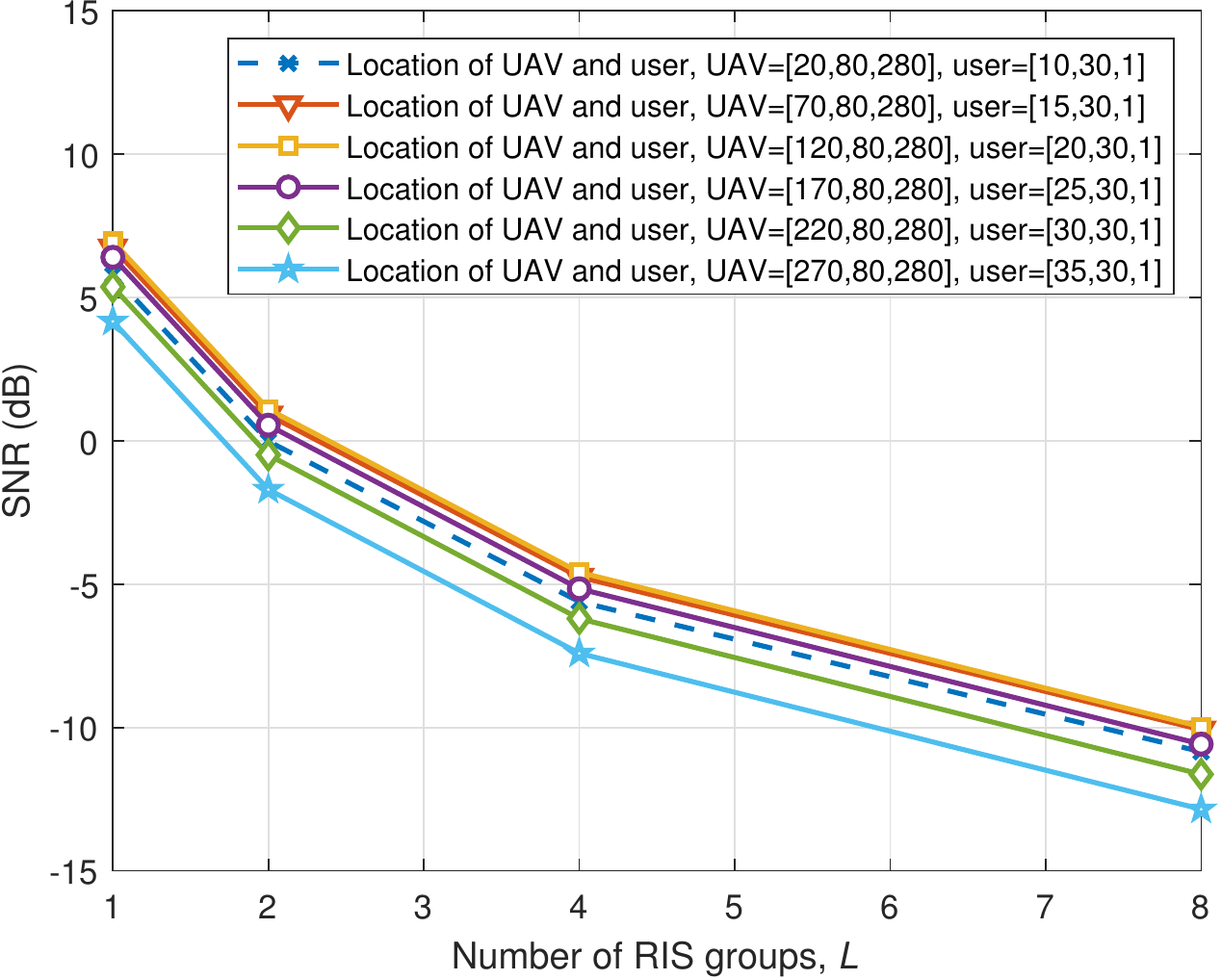}
    		\end{minipage}
    	}
    	\hfill
    	\subfigure[Transmit power vs. number of RIS groups]{
    		\begin{minipage}[b]{0.31\textwidth}
    			\includegraphics[width=1\textwidth]{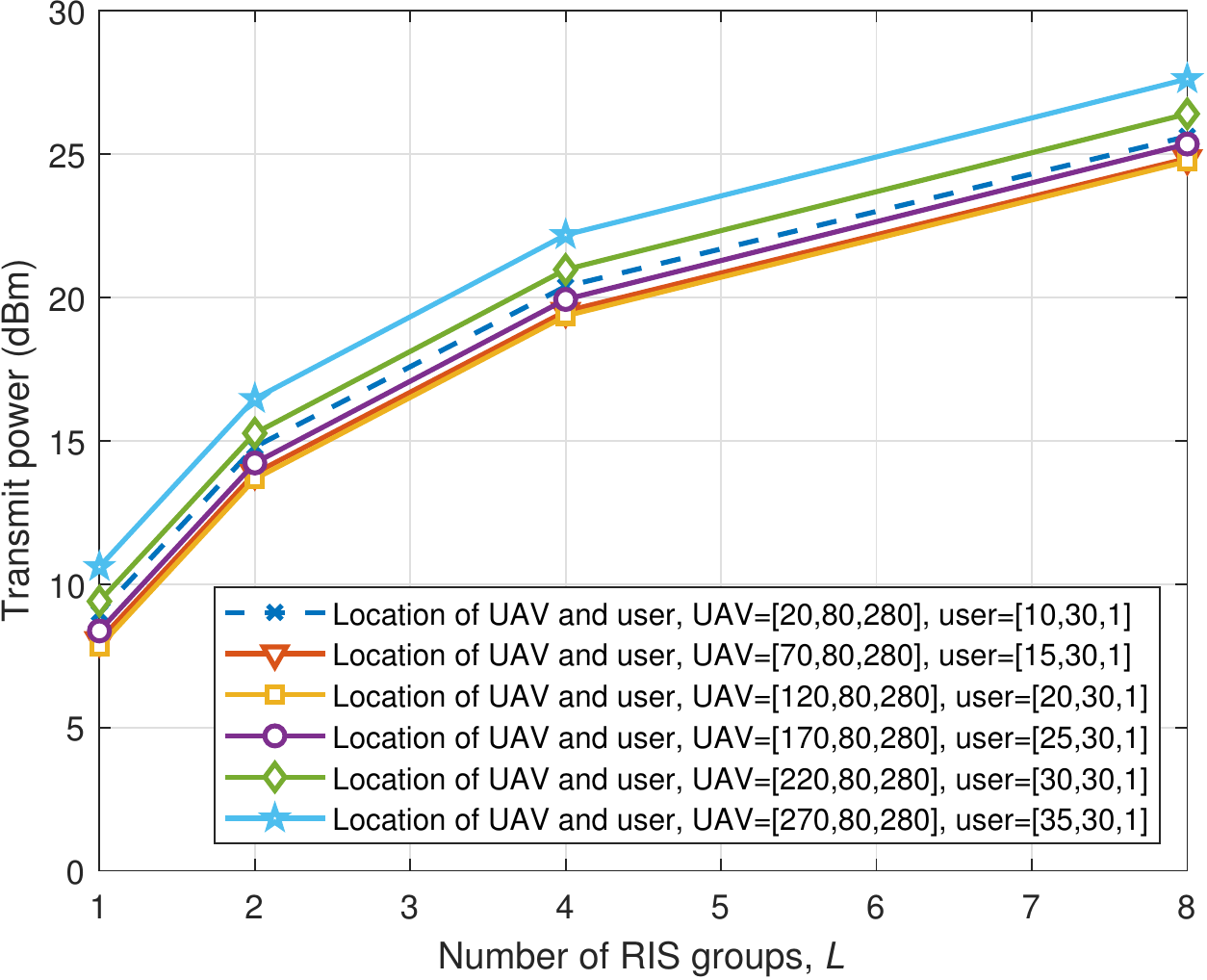}
    		\end{minipage}
    	}
    	\hfill
    	\subfigure[Throughput vs. number of UAV-user pairs]{
    		\begin{minipage}[b]{0.32\textwidth}
    			\includegraphics[width=1\textwidth]{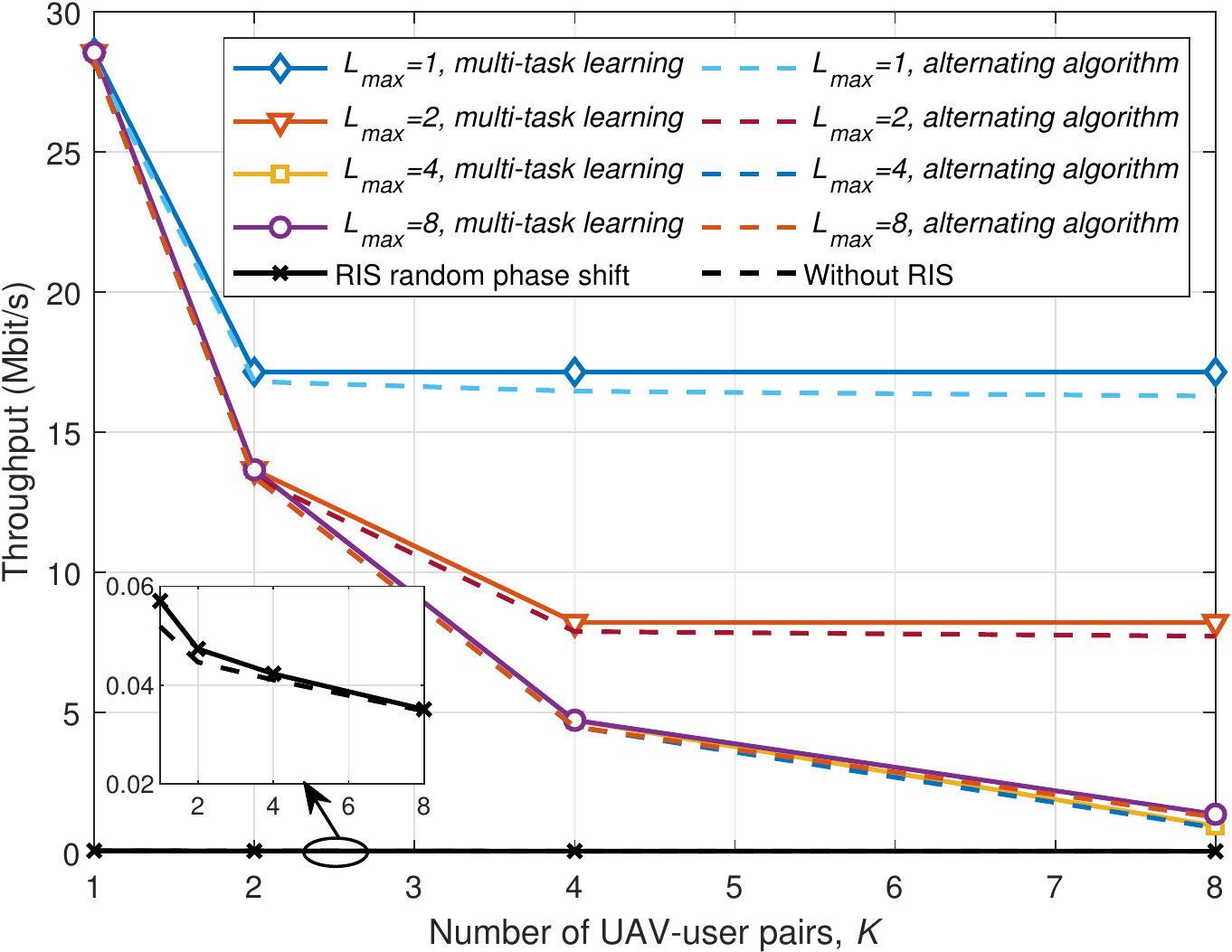}
    		\end{minipage}
    	}
    	\caption{SNR, transmit power, and throughput versus the number of RIS groups and UAV-user pairs.} 
    	\label{S1}
\end{figure*}

\begin{figure*}[t]
\small
      \centering
    	\subfigure[SNR vs. moving distance of the UAV-user pair]{
    		\begin{minipage}[b]{0.315\textwidth}
    			\includegraphics[width=1\textwidth]{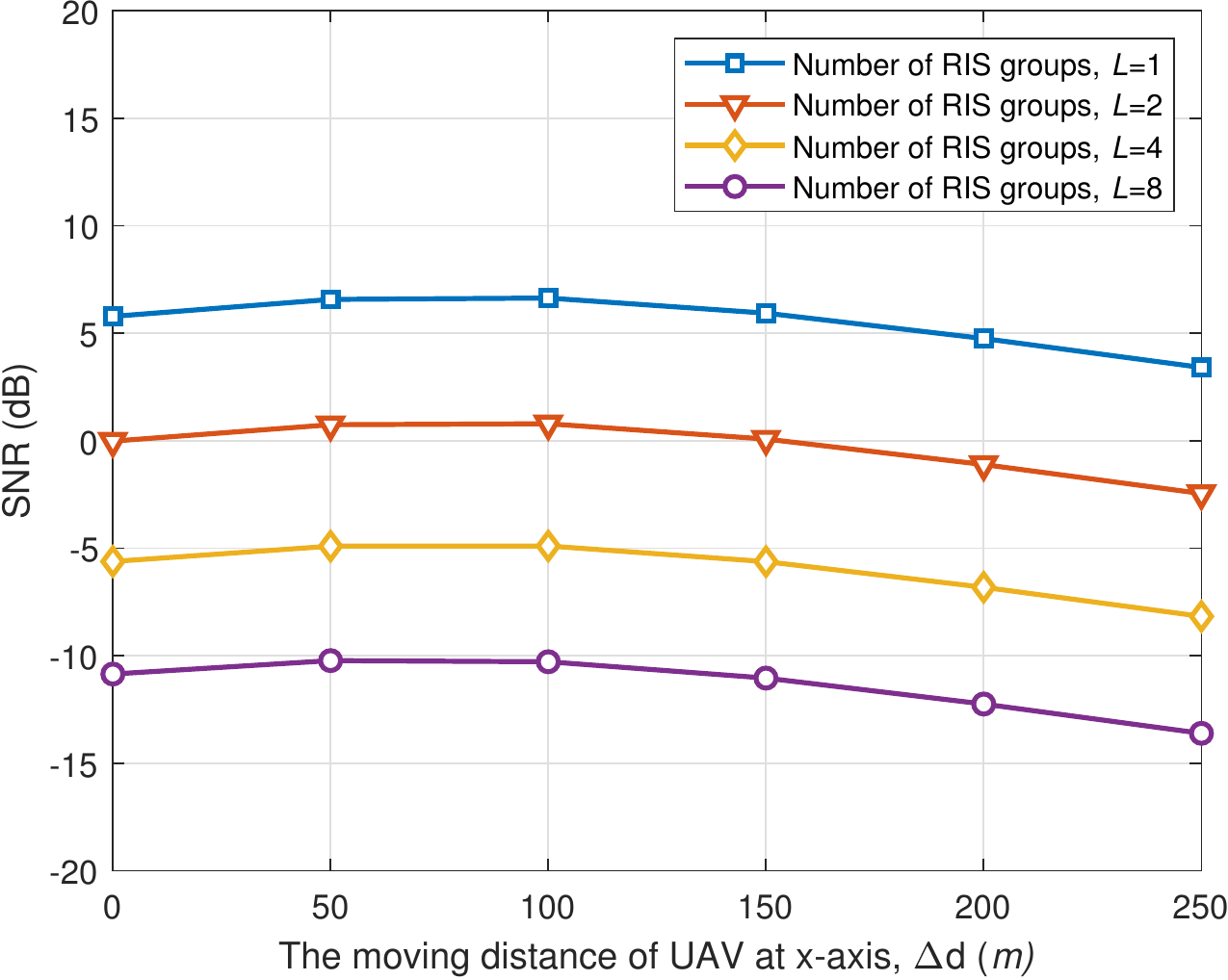}
    		\end{minipage}
    	}
    	\hfill
    	\subfigure[Transmit power vs. moving distance of the UAV-user pair]{
    		\begin{minipage}[b]{0.315\textwidth}
    			\includegraphics[width=1\textwidth]{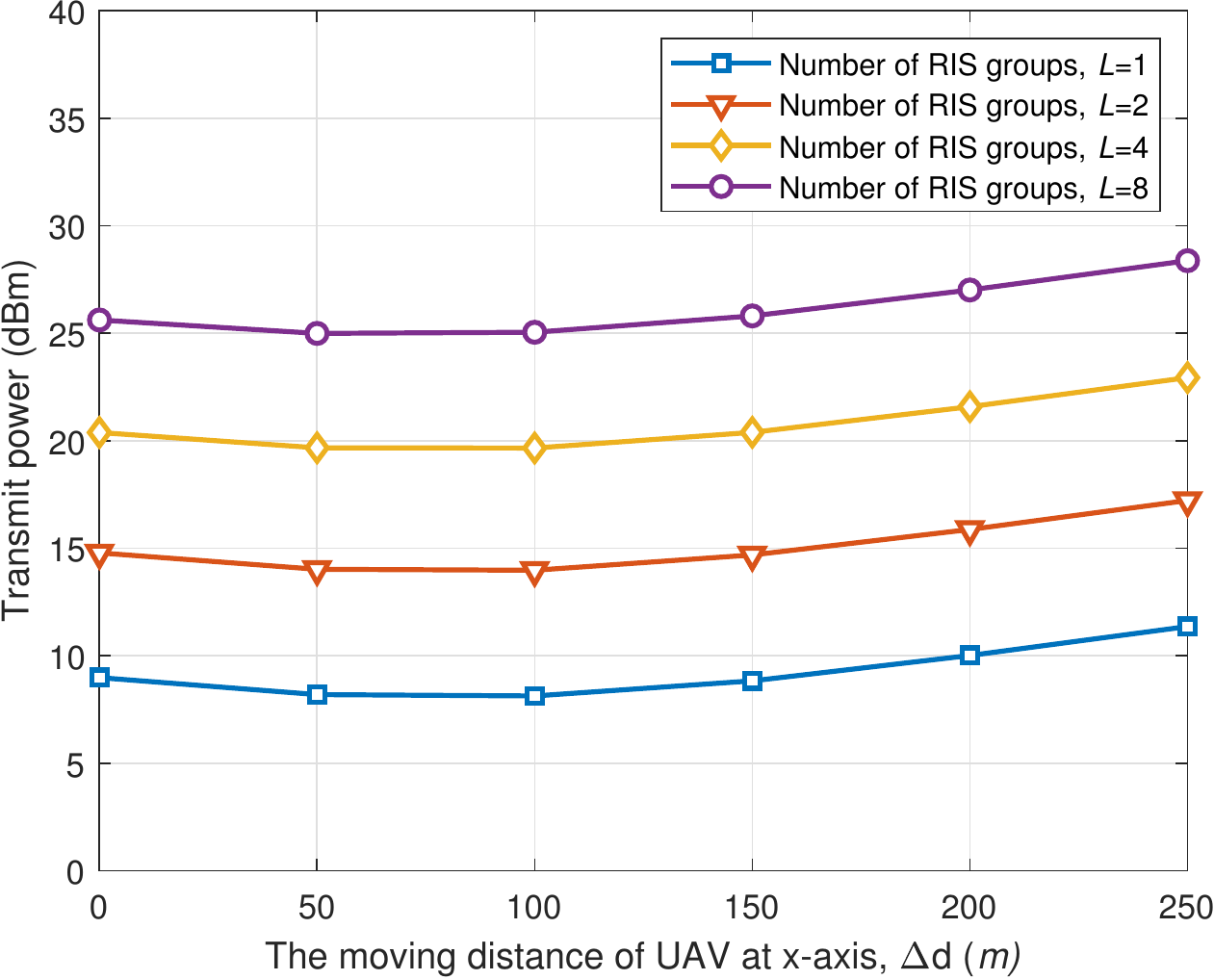}
    		\end{minipage}
    	}
    	\hfill
    	\subfigure[Throughput vs. moving distance of UAV-user pairs]{
    		\begin{minipage}[b]{0.315\textwidth}
    			\includegraphics[width=1\textwidth]{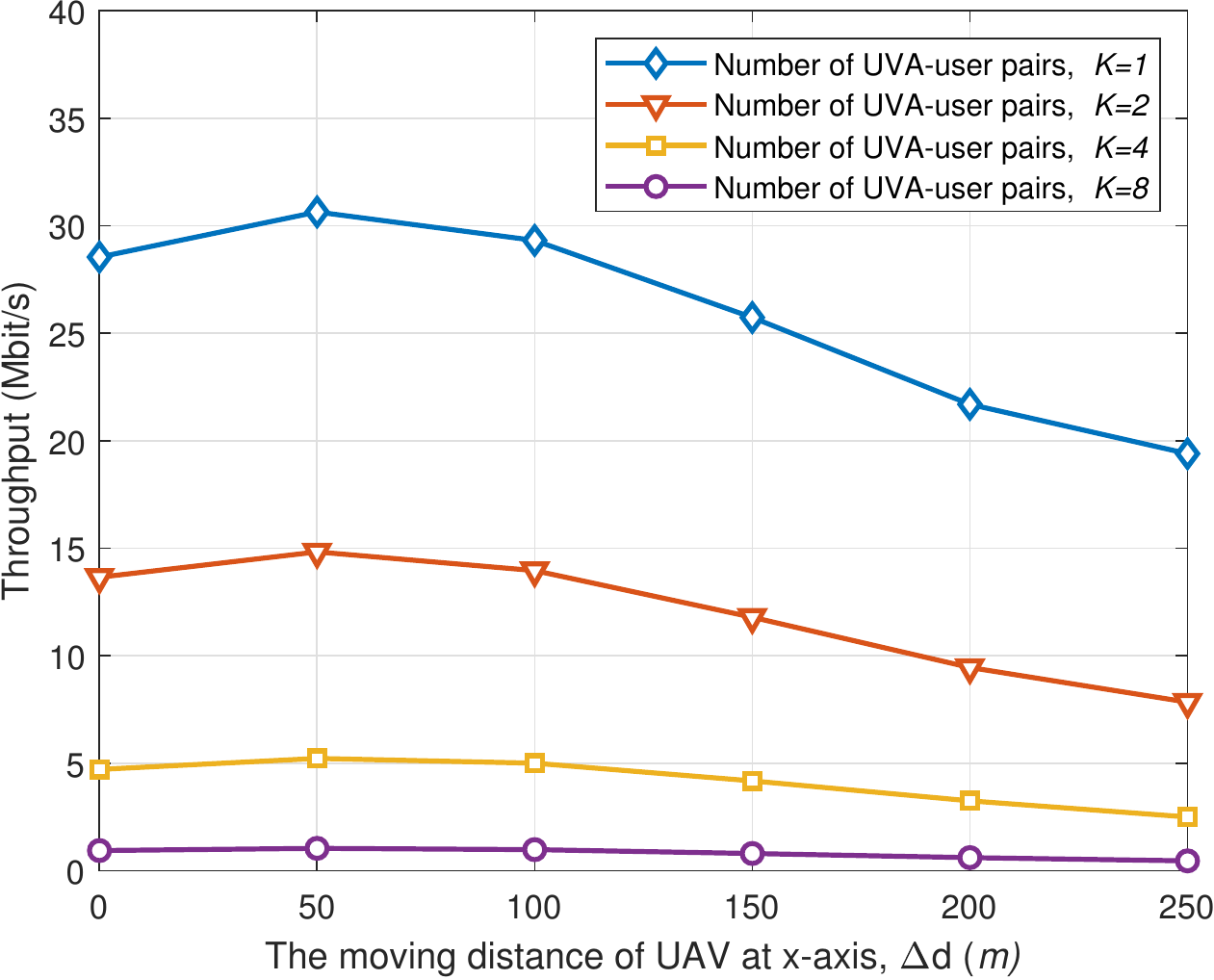}
    		\end{minipage}
    	}
    	\caption{SNR, transmit power, and throughput versus the moving distance of the UAV-user pair at x-axis along the same direction.} 
    	\label{S2}
\end{figure*} 

\begin{figure*}[t]
\small
      \centering
    	\subfigure[Accuracy vs. number of UAV-user pairs]{
    		\begin{minipage}[b]{0.41\textwidth}
    			\includegraphics[width=1\textwidth]{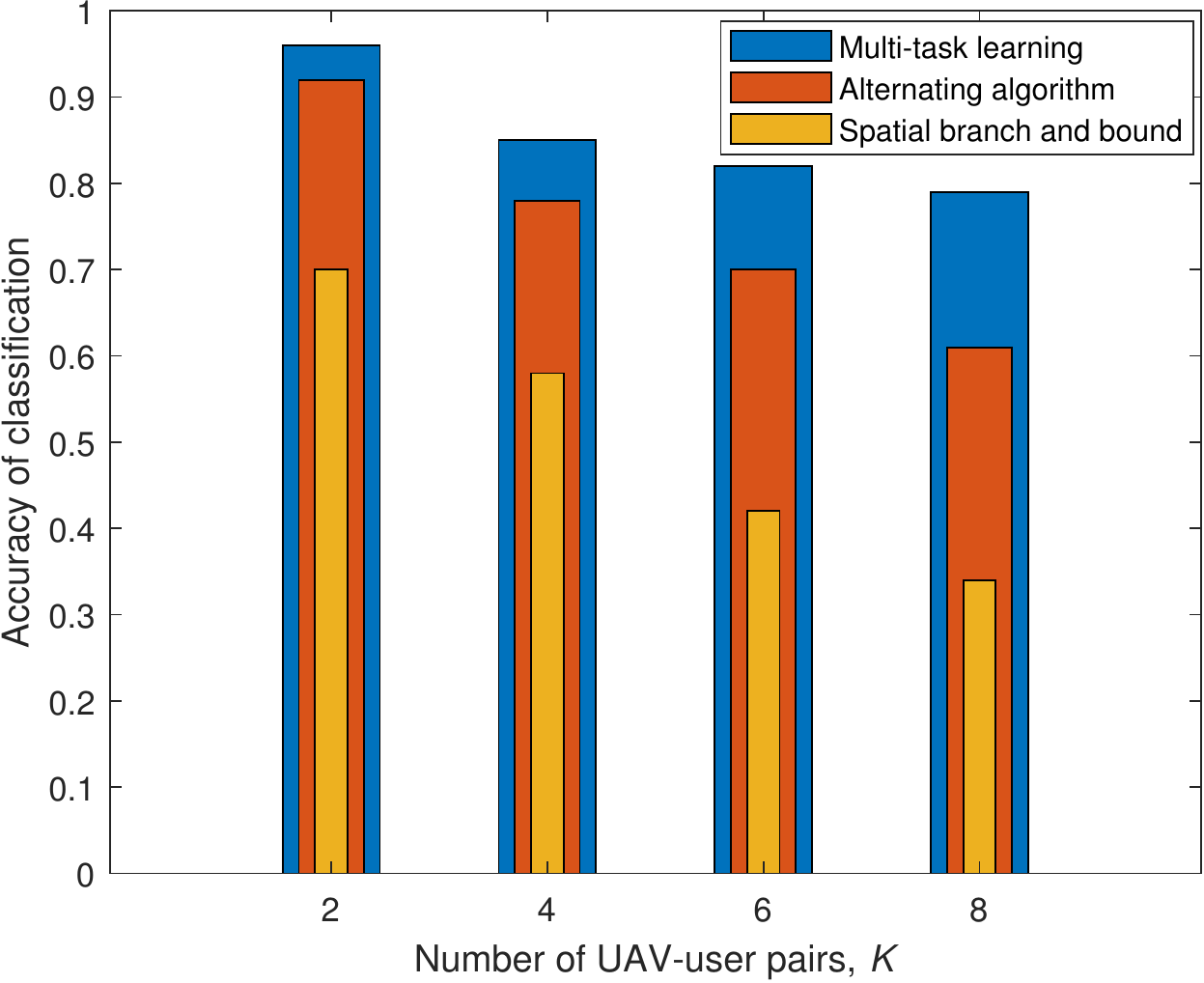}
    		\end{minipage}
    	}
    	\hfill
    	\subfigure[MSE vs. number of UAV-user pairs]{
    		\begin{minipage}[b]{0.41\textwidth}
    			\includegraphics[width=1\textwidth]{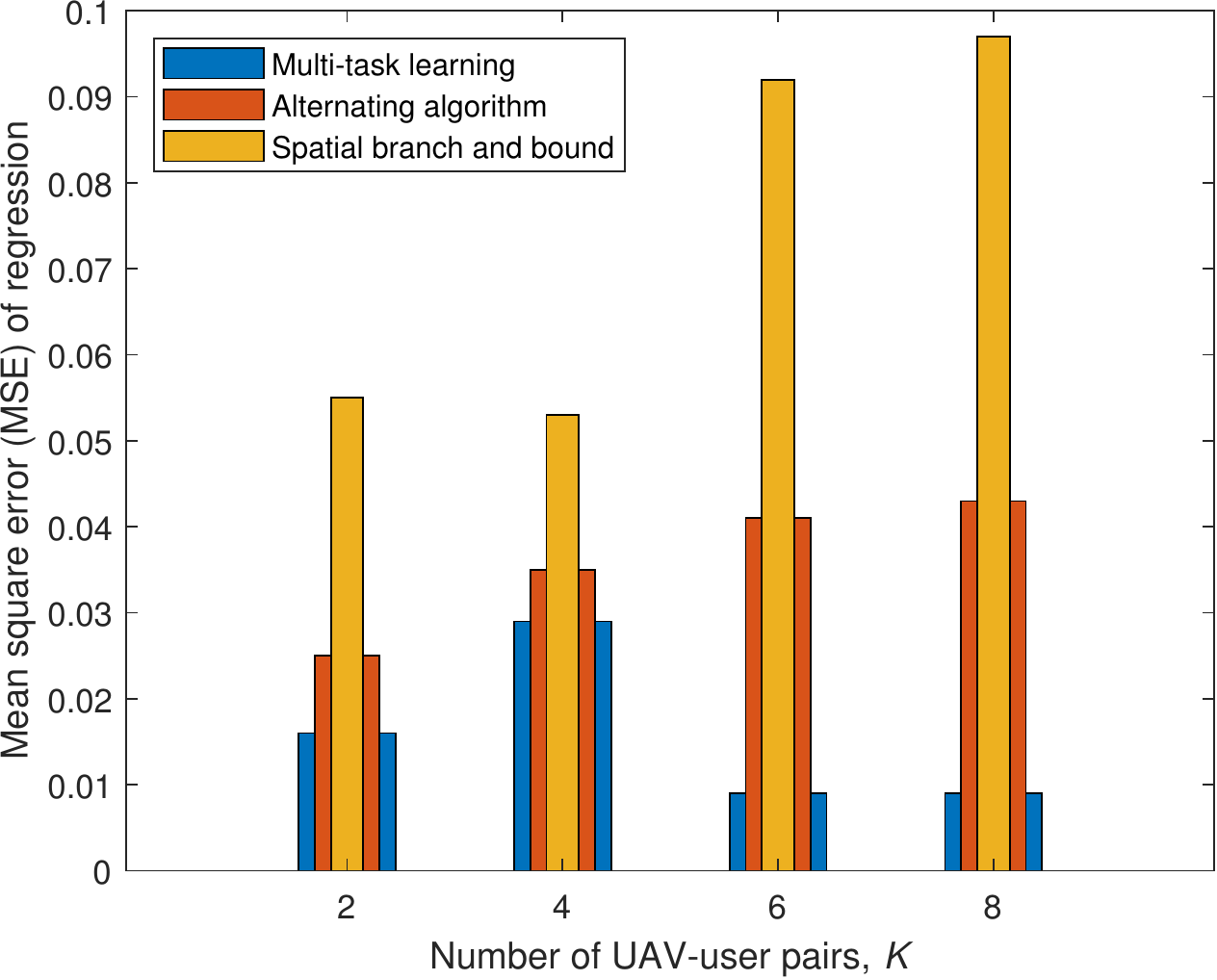}
    		\end{minipage}
    	}
    	\caption{The accuracy of classification and MSE of regression with multi-task learning, alternating algorithm, and spatial branch and bound.} 
    	\label{P1}
    \end{figure*} 
    
\begin{figure*}[t]
      \centering
    	\subfigure[Accuracy vs. percentage of training samples]{
    		\begin{minipage}[b]{0.41\textwidth}
    			\includegraphics[width=1\textwidth]{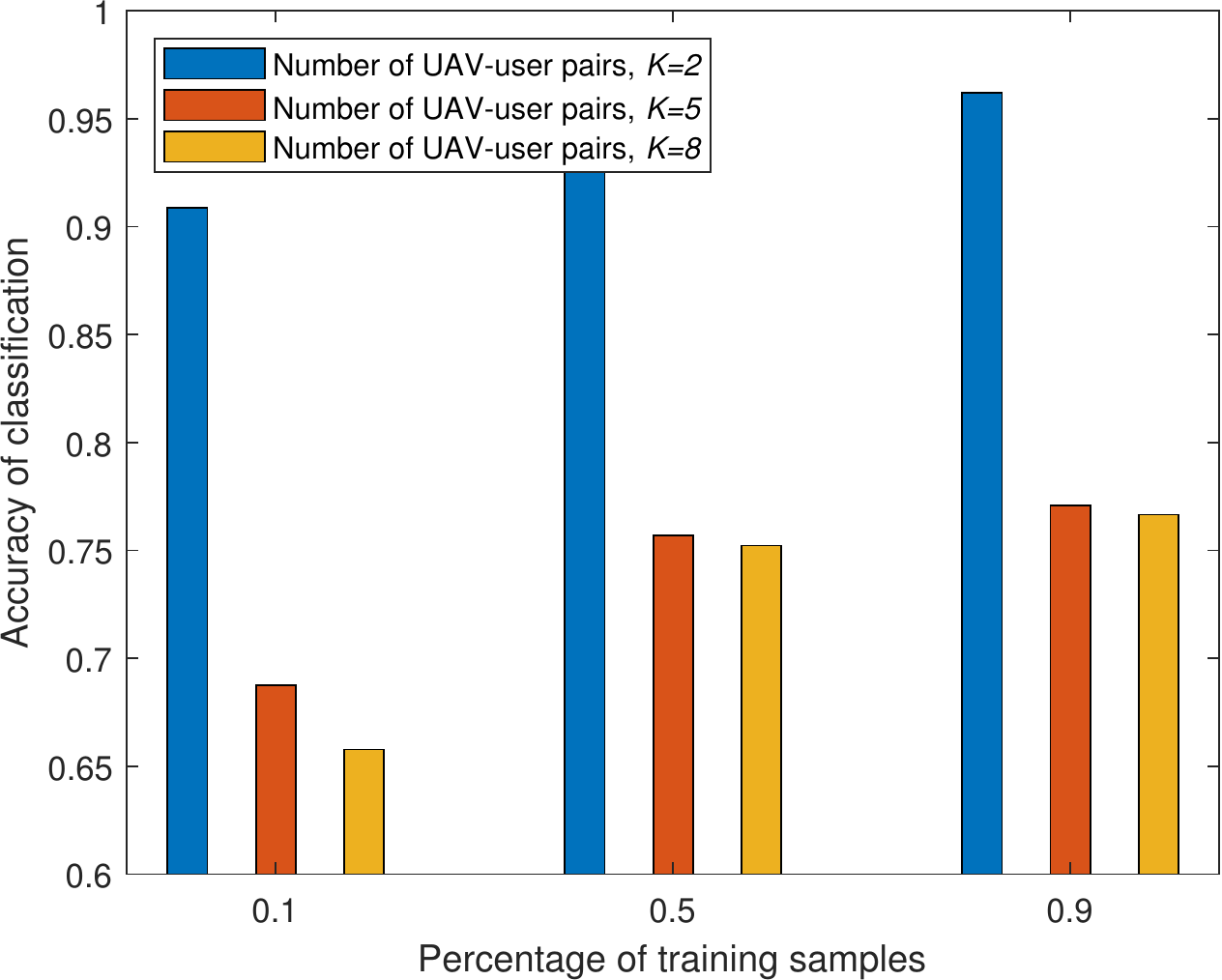}
    		\end{minipage}
    	}
        \hfill
    	\subfigure[MSE vs. percentage of training samples]{
    		\begin{minipage}[b]{0.41\textwidth}
    			\includegraphics[width=1\textwidth]{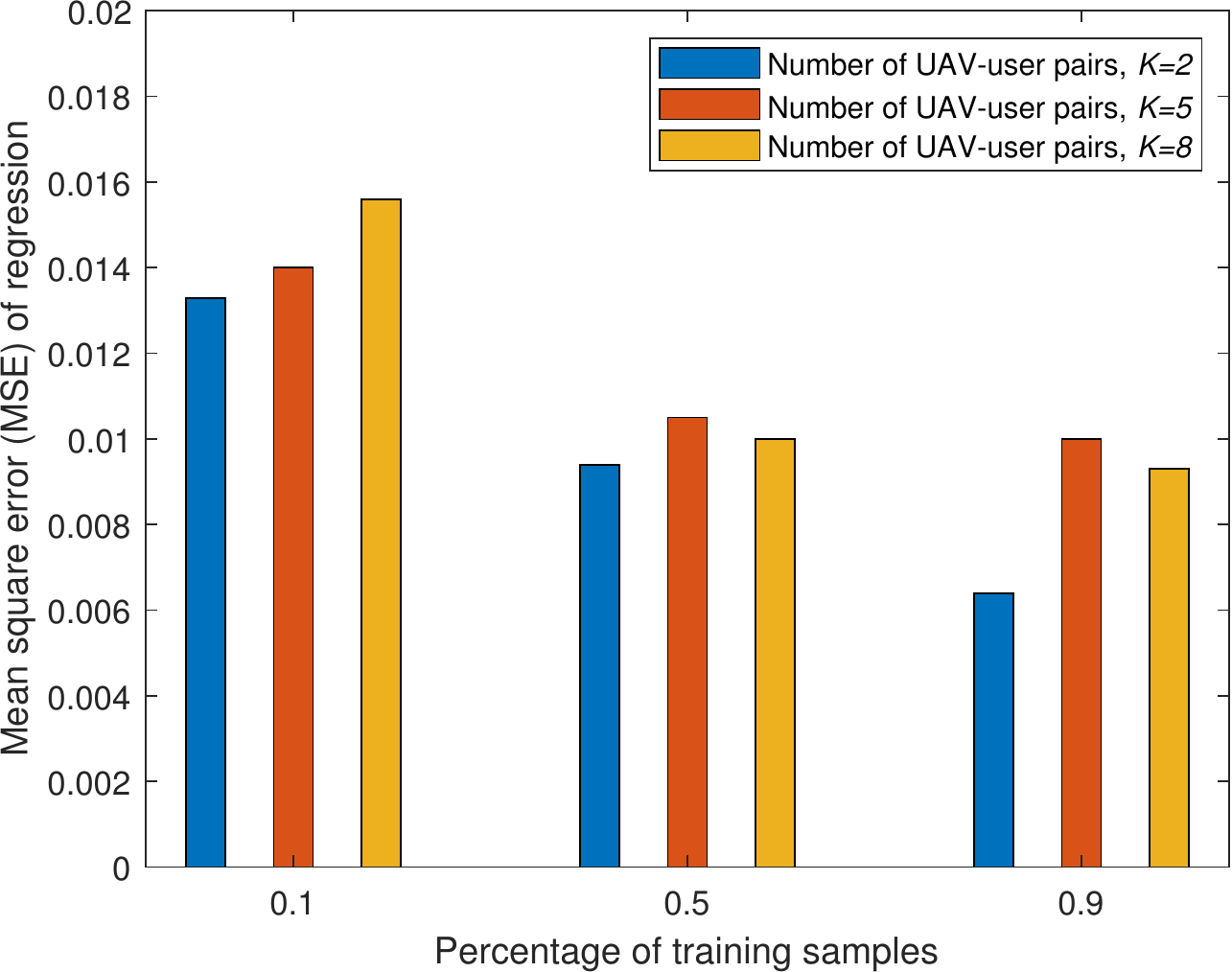}
    		\end{minipage}
    	}
    	\caption{The accuracy of classification and MSE of regression with the different number of UAV-user pairs.} 
    	\label{P2}
    \end{figure*}

Figure \ref{S2} evaluates the SNR, the transmit power, and the throughput as the UAV-user pairs move (i.e., UAVs move at the x-axis from 0 to 250 m, and users move at the x-axis from 0 to 25 m). Figure \ref{S2} (a) first shows that the SNR of Pair $\#$1 changes when the number of RIS groups is 1, 2, 4, and 8. We can see that the SNR has an apparent decrease as the number of RIS groups increases because the fewer RIS elements can be allocated to Pair $\#$1. Then, Fig. \ref{S2} (a) shows that the SNR varies as the moving distance ($\Delta d$) of Pair $\#$1 increases from 0 to 250 m, where the SNR shows a slight increase and then a minor decrease as the moving distance of Pair $\#$1 increases. Figure \ref{S2} (b) indicates that the transmit power changes as the number of RIS groups and the moving distance of Pair $\#$1. By observing Figs. \ref{S2} (a) and (b), it can be seen that an opposite changing trend happens on the transmit power compared to Fig. \ref{S2} (a), and we also see that the SNR is mainly affected by the size of the RIS groups, and then by the moving distance of Pair $\#$1. Finally, Fig. \ref{S2} (c) shows that the system throughput changes as the number of UAV-user pairs and the moving distance of the UAV-user pairs. With the moving distance of UAV-user pairs increasing from 0 to 250 m, the system throughput first slightly climb up and then has a minor descending. This is because that the UAV-user pairs are first close to the RIS and then gradually apart from the RIS. Besides, the system throughput decreases as the number of UAV-user pairs increases when the maximum number of the RIS groups is limited, which can be explained that the direct link of the UAV-user pairs without the assistance of the RIS severely reduce the performance of aerial-terrestrial communications system.

Figure \ref{P1} evaluates that the accuracy of the classification and the mean square error (MSE) of regression as the number of UAV-user pairs varies, where the alternating algorithm and the conventional spatial branch and bound (sBB) are regarded as the benchmarks to compare with the proposed MTL. Figure \ref{P1} (a) shows that the accuracy of classification varies with the number of UAV-user pairs for three methods. First, we can see that the accuracy of MTL classification outperforms the others as the number of UAV-user pairs increases. It is also observed that the classification accuracy of all methods decreases as the number of UAV-user pairs increases. Compared with the sBB and the alternating algorithm that decline significantly, MTL has only a minor decline. Figure \ref{P1} (b) shows that the MSE varies with the number of UAV-user pairs for three methods. It can be seen that the MSE of MTL outperforms the other two methods as the number of UAV-user pairs increases. Besides, Table \ref{IF} shows the inference time of each sample for sBB and MTL. Compared to sBB, the inference time of MTL significantly declines about four orders of magnitude of sBB. Therefore, MTL effectively speeds up the computation time with high inference accuracy, thereby improving the performance of RIS-assisted aerial-terrestrial communications.

Figure \ref{P2} evaluates the accuracy of the classification and the MSE of regression as the percentage of training samples and the number of UAV-user pairs vary. Figure \ref{P2} (a) shows that the accuracy of classification changes as the percentage of training samples when the number of UAV-user pairs is 2, 5, and 8. It can be seen that the accuracy of MTL classification is best at $k=2$, followed by $k=5$ and $k=8$, i.e., the accuracy declines as the number of UAV-user pairs increases. We also see that the accuracy of MTL classification increases as the percentage of training samples increases from $10\%$ to $90\%$, and the biggest gap of the accuracy exists when the number of UAV-user pairs is $8$. Figure \ref{P2} (b) shows that the MSE changes as the percentage of training samples and the number of UAV-user pairs. It can be seen that the MSE decreases as the percentage of training samples increases, and the biggest MSE gap is shown when the percentage of training samples is set $10\%$ and $90\%$ for two UAV-user pairs. 

\begin{table*}[t] 
		\small 
		\centering
			\renewcommand{\arraystretch}{1.2}
			\captionsetup{font={small}} 
			\caption{\scshape Inference Time Comparison} 
			\label{IF}
			%\footnotesize  
			\centering  
			\begin{tabular}{| c | c | c | c | c | c | c | c | c | }  
				\shline
                \diagbox{Method}{Inference time ($ms$)}{$K$} & 2 & 3 & 4 & 5 & 6 & 7 & 8\\
				\shline 
				sBB & 3.2  & 5.8 & 3.5 & 5.4 & 4.8 & 4.1 & 3.8 \\ 
			    \hline 
			    MTL & 0.0145 & 0.0146 & 0.0151 & 0.0155 & 0.0161 & 0.0168 & 0.0172\\ 
			    \shline
			\end{tabular}  
	\end{table*}

\section{Conclusion}\label{sec7}
In this paper, we investigated an RIS-assisted multi-user downlink aerial-terrestrial communication system. To improve the coverage and the link performance of UAV-user pairs, we proposed a frame-based RIS-assisted transmission protocol, especially when the LoS wireless links are blocked or seriously deteriorated. In the proposed protocol, the RIS controller can adaptively make a transmission strategy for all UAV-user pairs by addressing the following issues: the RIS elements allocation and the RIS phase shifts configuration. To address these two issues, we formulated a transmission strategy optimization problem to maximize the system capacity. To speed up a solution to the proposed MINLP problem, we trained and then deployed a multi-task learning model at the RIS controller instead of the conventional mathematical optimization methods. Simulation results show that the system throughput of RIS-assisted aerial-terrestrial communications is improved by hundreds-fold compared to that without the RIS. In addition, benefiting from multi-task learning, the computation time solving the MINLP at the RIS controller is decreased by four orders of magnitude compared to the conventional mathematics method, and the inference performance is also significantly improved. To realise the full potential of UAV's mobile controllability, the proposed RIS-assisted transmission strategy will be deeply explored considering the trajectory and deployment of UAVs in our future work.

\end{document}